\documentclass[twoside,a4paper]{IEEEtran} 
\usepackage[T1]{fontenc}
\usepackage[latin1]{inputenc}
\usepackage{amsmath,amssymb,amsthm,mathtools,bbold} %
\usepackage{amsmath,amssymb,amsthm,mathtools} %
\newcommand{\openone}{{\mathbb{1}}}
\newcommand{\eqnumber}{\refstepcounter{equation}\tag{\theequation}}

\usepackage{comment}
\usepackage[usenames,dvipsnames]{xcolor}
\usepackage{tikz}

\usepackage{xstring}

\usepackage{cleveref}
\usepackage{physics}

\newtheorem{theorem}{Theorem}
\newtheorem{lemma}       [theorem]{Lemma}
\newtheorem{corollary}   [theorem]{Corollary}
\newtheorem{proposition} [theorem]{Proposition}

\newtheorem{definition}  [theorem]{Definition}

\newtheorem{remark}      [theorem]{Remark}

\newtheorem{conjecture}  [theorem]{Conjecture}
\newtheorem{fact}        [theorem]{Fact}
\newtheorem{construction}[theorem]{Construction}

\newtheorem*{note*}{Note}

\newcommand{\N}{\mathbf{N}}

\newcommand{\R}{\mathbf{R}}
\newcommand{\C}{\mathbf{C}}

\renewcommand{\leq}{\leqslant}
\renewcommand{\geq}{\geqslant}
\newcommand{\st}{\  : \ } 

\DeclareMathOperator{\vrad}{vrad}
\DeclareMathOperator{\conv}{conv}

\DeclareMathOperator{\id}{id}
\DeclareMathOperator{\E}{\mathbf{E}}

\newcommand{\probability}{\mathbf{P}\quantity}
\newcommand{\robustness}{\mathcal{R}}
\newcommand{\GUE}{G}

\newcommand{\ppt} {{\mathbf{PPT}}}
\newcommand{\sep} {{\mathbf{SEP}}}
\newcommand{\locc}{{\mathbf{LOCC}}}

\newcommand{\states}{\mathcal{K}}
\newcommand{\allstates}{\mathcal{D}}
\newcommand{\pptstates}{\mathcal{P}}
\newcommand{\sepstates}{\mathcal{S}}

\newcommand{\channels}{{\mathbf{L}}}
\newcommand{\pptchannels} {{\ppt}}

\newcommand{\measurement}{{\mathcal{M}}}
\newcommand{\measurements}{{\mathbf{M}}}
\newcommand{\allmeasurements}{{\mathbf{ALL}}}
\newcommand{\pptmeasurements} {{\ppt}}

\newcommand{\loccmeasurements}{{\locc}}

\newcommand{\shield         }{\varrho}
\newcommand{\shieldplus     }{\shield^+}
\newcommand{\shieldminus    }{\shield^-}
\newcommand{\shieldplusminus}{\shield^\pm}

\newcommand{\shieldpm}{\shieldplusminus}
\newcommand{\randomshield         }{\rho}
\newcommand{\randomshieldplus     }{\randomshield^+}
\newcommand{\randomshieldminus    }{\randomshield^-}
\newcommand{\randomshieldplusminus}{\randomshield^\pm}

\newcommand{\fixedshield         }{\bar\randomshield}
\newcommand{\fixedshieldplus     }{\fixedshield^+}
\newcommand{\fixedshieldminus    }{\fixedshield^-}
\newcommand{\fixedshieldplusminus}{\fixedshield^\pm}

\newcommand{\lowerconstant}{c} 
\newcommand{\upperconstant}{C} 
\newcommand{\probconstant}{c_0} 

\newcommand{\system}[1]{{\mathrm{#1}}}
\newcommand{\classic}[1]{{\mathrm{#1}}}

\newcommand{\quantum}[1]{\mkern2mu\underline{\mkern-2mu\mathrm{#1}\mkern-2mu}\mkern2mu }

\newcommand{\rA}{\mathrm{A}}
\newcommand{\rB}{\mathrm{B}}

\date{\today}

\begin{document}
\title{Random private quantum states}

\author{Matthias Christandl, Roberto Ferrara, C\'{e}cilia Lancien
    \thanks{Matthias Christandl is with QMATH, Department of Mathematical Sciences, University of Copenhagen, 2100 Copenhagen, Denmark.}
    \thanks{Roberto Ferrara is currently with Lehr- und Forschungseinheit f\"ur Nachrichtentechnik, Technische Universit\"at M\"unchen, Munich, Germany.  Part of the work was carried out, while he was at QMATH, Department of Mathematical Sciences, University of Copenhagen, 2100 Copenhagen, Denmark.}
    \thanks{C\'{e}cilia Lancien is with Institut de Mathématiques de Toulouse \& CNRS, F-31062 Toulouse Cedex 9, France. Part of the work was carried out, while she was at Departamento de An\'{a}lisis Matem\'{a}tico, Universidad Complutense de Madrid, 28040 Madrid, Spain \& Instituto de Ciencias Matem\'{a}ticas, 28049 Madrid, Spain.}
\thanks{This paper was presented in part at ISIT'18.}

}

\maketitle

\begin{abstract}
The study of properties of randomly chosen quantum states has in recent years led to many insights into quantum entanglement. In this work, we study private quantum states from this point of view. Private quantum states are bipartite quantum states characterised by the property that carrying out simple local measurements yields a secret bit. This feature is shared by the maximally entangled pair of quantum bits, yet private quantum states are more general and can in their most extreme form be almost bound entangled. In this work, we study the entanglement properties of random private quantum states and show that they are hardly distinguishable from separable states and thus have low repeatable key, despite containing one bit of key. The technical tools we develop are centred around the concept of locally restricted measurements and include a new operator ordering, bounds on norms under tensoring with entangled states and a continuity bound for a relative entropy measure. 

\end{abstract}


\section{Introduction}
\label{sec:intro}

The study of random quantum states with probabilistic tools and high dimensional analysis has in recent years significantly advanced our understanding of entanglement, the strong quantum correlations present in quantum systems~\cite{HLW,ASY,AL}. In this work, we use such techniques in order to construct bipartite quantum states that exhibit a large gap between, on the one hand, their key distillation properties and, on the other hand, their entanglement distillation and key repeater distillation properties. 

In order to do so, we follow the prescription of~\cite{HHHO,HHHOb} to construct bipartite quantum states that contain a readily accessible bit of pure privacy, so-called private quantum states. These are constructed as follows.
We give Alice and Bob a Bell state, $\psi^+_{\system{AB}}$, or the Bell state subject to a phase flip, $\psi^-_{\system{AB}}$, with probability one half.
Then, we store the information of whether or not a phase flip has been applied in a pair of orthogonal shield states $\shieldplusminus_{\system{A}'\system{B}'}$. This results in the \emph{private quantum state}
$$\gamma_{\system{A}\system{A}'\system{B}\system{B}'} \coloneqq 
\frac{1}{2} \psi^+_{\system{AB}}\otimes \shieldplus_{\system{A}'\system{B}'} + 
\frac{1}{2} \psi^-_{\system{AB}}\otimes \shieldplus_{\system{A}'\system{B}'}.$$
It can be shown that the bit that Alice and Bob obtain by measuring $\system{A}\otimes\system{B}$ in the computational basis is secret, since the shield states are orthogonal, and an eavesdropper with access to the purification cannot erase this information.  

Assume now that the shield states are data hiding~\cite{DVLT}, meaning that Alice and Bob can barely distinguish them if they are only able to perform LOCC measurements on them. It is then the intuition that Alice and Bob cannot distill entanglement, since they have poor access to the phase information to be corrected in the entanglement distillation process. In our case, we choose the data hiding states at random and in high dimension~\cite{AL}. We are able to prove that, with high probability, the private quantum states have distance from separable states which is high when measured in the PPT-restricted norm or PPT-restricted relative entropy distances, yet low when measured in the SEP-restricted norm or SEP-restricted relative entropy distances. 

We then consider the quantum repeater scenario~\cite{BCHW}, in which Alice and Bob are connected via one intermediary repeater station (Charlie). We distribute a random private state between Alice and Charlie and another one between Charlie and Bob. We then show that, despite the fact that the private states contain one bit of readily extractable secrecy, any repeater protocol (with the repeater station limited to single copy operations) will fail to extract secrecy: the quantum key repeater rate is vanishing for large dimensions. This goes beyond the constructions in previous works, where upper bounds were always derived for states that are data hiding under PPT measurements, something that is excluded in our construction.
Notice that the states between adjacent nodes in a network will generally be specifically designed states rather than random states.
However, our results point at a potentially important pitfall to be aware of
in the implementation of real QKD networks.
A network might have a good key rate between adjacent nodes
and have good operations at the repeater stations,
but this is not enough to guarantee a good key rate between distant nodes.
Our results are another step pointing toward the distillable entanglement being the only relevant resource for repeating quantum information.
If this turned out to be true, then small deviations from the designed distributed states might have a large effect on the key rate between non-adjacent nodes.

\smallskip

The paper is structured as follows. In \Cref{sec:EPR}, we define several notions of measurement-restricted distance measures, on which we establish several kinds of bound (continuity, increase under tensoring, etc.). These technical statements are crucial in our subsequent study of entanglement properties of private quantum states, but might also be of independent interest. In \Cref{sec:p-bit}, we then introduce our model of random private quantum states. We use the results proved in \Cref{sec:EPR}, together with concentration of measure techniques, to establish bounds on their typical distinguishability from separable states. This brings us to \Cref{sec:repeater}, where the main result of our paper appears, as Theorem \ref{th:random-R1}. It consists of a bound on an adapted quantum key repeater rate for random private quantum states. It is proved by first upper-bounding this quantum key repeater rate of interest in terms of some of the previously studied distinguishability measures, so that we can apply the results of \Cref{sec:p-bit} to conclude. We also discuss the relation of this work to the PPT\textsuperscript{2} conjecture, which was motivated by the key repeater scenario. We conclude in \Cref{sec:remarks} with a discussion (on differences between our construction and previous ones, on the choice of randomness in our work, etc.). In Appendix \ref{section:iso} we discuss the local distinguishability of isotropic states, as an additional observation related to the topic of \Cref{sec:EPR}. In Appendix \ref{appendix:improvement}, we introduce a new notion, that of measurement-restricted operator ordering relation. The latter allows to obtain slightly better data-hiding bounds for random data-hiding quantum states than the ones of \Cref{sec:p-bit}. In Appendix \ref{appendix:R_D}, we present a slight improvement on previous quantum key repeater bounds, and then discuss a reformulation of this bound for private states.

\section{Locally restricted distinguishability measures}
\label{sec:EPR}

\subsection{Restricted norm distances}  
\label{sec:Mnorms}

Let $\system{H}$ be a complex Hilbert space, which we always take to have finite dimension. 
On the set of Hermitian operators on $\system{H}$, we define $\|\cdot\|_1$ as the trace norm, $\|\cdot\|_2$ as the Hilbert--Schmidt norm, $\|\cdot\|_{\infty}$ as the operator norm, and $B_1$, $B_2$, $B_{\infty}$ as the corresponding unit balls. 
Given $S$ a subset of the Hermitian operators on $\system{H}$, we denote by $\conv(S)$ the convex hull of the elements of $S$ and by $\overline{\conv}(S)$ the closure of $\conv(S)$. 
Given $\mathcal{C}$ a symmetric convex subset of the Hermitian operators on $\system{H}$, we denote by $\norm{X}_{\mathcal{C}}\coloneqq\inf\left\{t: X\in t\mathcal{C} \right\}$ its gauge (a norm), also known as Minkowski's functional, and by $\mathcal{C}^\circ \coloneqq \{Y: \forall\ X\in\mathcal{C},\ \Tr(XY)\leq 1\}$ its polar.

The set of all quantum states (or density operators) on $\system{H}$ is defined as the set of trace $1$ positive semidefinite operators on $\system{H}$, and is denoted by $\allstates\equiv\allstates(\system{H})$. Given $\mathcal{K}$ a convex subset of $\allstates$, we define $\R^+\mathcal{K}\coloneqq \{\lambda\varrho \st \lambda\geq 0,\ \varrho\in\mathcal{K}\}$ as the cone generated by $\mathcal{K}$.

A measurement on $\system{H}$ is characterized by a finite collection of Hermitian operators ${(T_i)}_{i\in I}$ on $\system{H}$ 
such that $\sum_{i\in I}T_i =\openone$ and $T_i\geq 0$ for each $i\in I$. It is therefore often referred to as a positive operator-valued measure (POVM).
One can equivalently associate to any such measurement ${(T_i)}_{i\in I}$ 
the quantum-to-classical channel $\measurement$ 
that maps any Hermitian operator $X$ on $\system{H}$ to $\measurement(X)\coloneqq \sum_{i\in I}\Tr(T_{i} X)\ketbra{i}{i}$.
We denote by $\allmeasurements\equiv\allmeasurements(\system{H})$, the set of all measurements on $\system{H}$.
\smallskip

Given a set of measurements $\measurements\subseteq\allmeasurements$ on $\system{H}$, we have the following notion of distinguishability in restriction to $\measurements$, which will be crucial throughout the whole paper. 
 
\begin{definition}[$\measurements$ norm~\cite{MWW}] 
\label{def:M-norm}
Let $\measurements$ be a set of measurements on $\system{H}$. 
For any Hermitian operator $X$ on $\system{H}$, its trace norm in restriction to $\measurements$, or $\measurements$ norm, is defined as: 
\begin{equation} 
\label{eq:M-norm}
\|X\|_{\measurements} \coloneqq  \sup_{\measurement\in\measurements}\|\measurement(X)\|_1\, . 
\end{equation}
Let $\states$ be any set of states and $\varrho$ be any state on $\system{H}$, 
then its trace norm distance from $\states$ in restriction to $\measurements$, or $\measurements$ norm distance from $\states$, is defined as:
\[ \|\varrho-\states\|_{\measurements} \coloneqq \inf_{\varsigma\in\states} \|\varrho-\varsigma\|_{\measurements}  \, . \]
The unrestricted norm distance of $\varrho$ from $\states$ is defined as:
\[ \|\varrho-\states\|_{1} \coloneqq \inf_{\varsigma\in\states} \|\varrho-\varsigma\|_{1}  \, . \]
\end{definition}
If $\measurements$ is such that $\norm{X}_\measurements=0$ if and only if $X=0$,
then the $\measurements$ norm is indeed a norm. In this case, $\measurements$ is often referred to as being ``informationally complete''. This will be the case for all the sets of measurements that we will consider in this paper.

The $\measurements$ norm can be always expressed in the following convenient form~\cite{MWW}:
\begin{equation}
\label{eq:norm-gauge}
 \|X\|_{\measurements}=\sup \{ \Tr(TX): T\in K_{\measurements} \} \, ,
\end{equation}
where $K_{\measurements}\subseteq B_{\infty}$ is the symmetric body defined as 
\begin{equation*}
K_{\measurements} \coloneqq
\overline{\conv}\quantity\big{2M-\openone: (M,\openone-M)\in\measurements}
\,.
\end{equation*}
By construction then,  $K_\measurements^\circ$ is the unit ball for $\|\cdot\|_{\measurements}$ and we have $\|\cdot\|_{K_\measurements^\circ} = \|\cdot\|_{\measurements}$.
If there exists a positive semidefinite closed convex cone $\R^+\states$ on $\system{H}$ that generates $\measurements$, namely such that
\begin{equation}
    \label{eq:measurements-cone}
    \measurements = \quantity{ {(T_i)}_{i\in I} \in \allmeasurements : \forall\ i\in I,\ T_i\in\R^+\states } \, ,
\end{equation}
then the symmetric convex body $K_{\measurements}$ simplifies to
\begin{align}
K_{\measurements}
  &=\quantity\big{2M-\openone: (M,\openone-M)\in\measurements}
  \nonumber
\\&=\quantity{\R^+\states-\openone}\cap \quantity{\openone - \R^+\states } \, . 
\label{eq:M-convexcone}
\end{align}
Notice that if $\measurements=\allmeasurements$, then we recover the trace norm:
\begin{equation}
\label{eq:Mnorm-tracenorm}
\norm{\cdot}_\allmeasurements = \norm{\cdot}_1\, .
\end{equation}
In other words, the trace norm can always be achieved by a measurement.
See~\cite{MWW} for further details.
\smallskip

A similar framework for studying the distinguishability in restriction to measurements, but using the relative entropy instead of the trace norm, was introduced in~\cite{Piani}. Based on the definition of the relative entropy of $\varrho$ and $\varsigma$
\[D(\varrho\|\varsigma)\coloneqq \Tr[\varrho(\log\varrho-\log\varsigma)]\, .\]
one defines, like in \Cref{def:M-norm},  the $\measurements$ relative entropy
\[D_{\measurements}(\varrho\|\varsigma) \coloneqq  \sup_{\measurement\in\measurements} D(\measurement(\varrho)\|\measurement(\varsigma)) \, ,\]
and the corresponding $\measurements$ relative entropy distance from $\states$ 
\[ D_{\measurements}(\varrho\|\states) \coloneqq \inf_{\varsigma\in\states} D_{\measurements}(\varrho\|\varsigma) \,. \]
However, we will later need more general sets than a set of measurements $\measurements$, therefore we introduce these definitions in the next section. The reason for needing more than measurements is that, for the relative entropy the analogue of \Cref{eq:Mnorm-tracenorm} does not hold. Namely, we cannot recover the unrestricted relative entropy just by computing the relative entropy restricted to all measurements. Indeed, $D_\allmeasurements(\varrho\|\varsigma) = D(\varrho\|\varsigma)$ if and only if $\varrho$ and $\varsigma$ commute~\cite{berta2016variational}, otherwise  $D_\allmeasurements(\varrho\|\sigma) < D(\varrho\|\sigma)$ (by monotonicity of the relative entropy).

\smallskip

Finally, let us emphasize that, in restricting to measurements, the trace norm and the relative entropy are equivalent to their classical counterpart on the (classical) measurement outcome.
Indeed, in \Cref{eq:M-norm} and \Cref{eq:M-D} we can rewrite $\|\cdot\|_1$ and $D(\cdot\|\cdot)$ as classical $1$-norm and classical relative entropy (also known as statistical distance and Kullback-Leibler divergence, respectively), i.e.
\begin{align*}
\|\measurement(X)\|_1 &=\sum_{i\in I}|\Tr(T_{i} X)|\,,
\\
D(\measurement(\varrho)\|\measurement(\varsigma)) &=
\sum_{i\in I}\Tr(T_i\varrho)
\quantity\big[\log\Tr(T_i\varrho)-\log\Tr(T_i\varsigma)]\,.
\end{align*}

\subsection{Restricted relative entropy distances} 

The generalization of \Cref{def:M-norm} to sets of quantum channels which are not necessarily quantum-to-classical is as follows: let $\channels$ be a set of quantum channels on $\system{H}$. With this we mean a channel from $\system{H}$ to $\system{H'}$, where $\system{H'}$ might have a dimension different from $\system{H}$. Define
\begin{align*}
\|X\|_{\channels} \coloneqq&  \sup_{\Lambda\in\channels}\|\Lambda(X)\|_1\, ,
\\
\|\varrho-\states\|_{\channels} \coloneqq& \inf_{\varsigma\in\states} \|\varrho-\varsigma\|_{\channels}  \, .
\end{align*}
By monotonicity of the trace norm, if $\channels$ contains the identity channel, then the latter is always the optimal channel.
In such case the above definition is not very interesting. But in \Cref{sec:repeater} we will need to consider sets of partial measurements; namely, we will require some subsystems to be measured but not others (see \Cref{sec:bipartite} for precise definitions and concrete examples). These sets of channels exclude the identity, but also include channels that are not measurements.
However, introducing an $\channels$ norm remains, by itself, of little interest anyway. Indeed, for any set of channels $\channels$, if $\measurements$ is the set of measurements obtained by composing any measurement in $\allmeasurements$ with any channel in $\channels$, then by \Cref{eq:Mnorm-tracenorm} we have:
\begin{equation}
\label{eq:Mnorm-Lnorm}
\|X\|_{\measurements} = \|X\|_{\channels}\, .
\end{equation}
Nonetheless, this equation will be a useful technical tool in our results.
\smallskip

The necessity of defining restricted distinguishability going beyond sets of measurements will appear clearer for the relative entropy than for the trace norm. Here again, the interesting cases will be those where the considered set of channels does not contain the identity.

\begin{definition}[$\channels$ relative entropy~\cite{CF}]
    \label{def:M-D}
    Let $\channels$ be a set of quantum channels on $\system{H}$. 
    For any states $\varrho$ and $\varsigma$ on $\system{H}$,
    their relative entropy in restriction to $\channels$, or $\channels$ relative entropy, is defined as:
    \begin{equation} \label{eq:M-D} 
    D_{\channels}(\varrho\|\varsigma) \coloneqq  \sup_{\Lambda\in\channels} D(\Lambda(\varrho)\|\Lambda(\varsigma))\,. 
    \end{equation}
    Let $\states$ be any set of states and $\varrho$ be any state on $\system{H}$, then its relative entropy from $\states$ in restriction to $\channels$,
    or $\channels$ relative entropy from $\states$, is defined as:
    \[ D_{\channels}(\varrho\|\states) \coloneqq \inf_{\varsigma\in\states} D_{\channels}(\varrho\|\varsigma) \,. \]
    The unrestricted relative entropy of $\varrho$ from $\states$ is defined as~\cite{DH}:
    \[ D(\varrho\|\states) \coloneqq \inf_{\varsigma\in\states} D(\varrho\|\varsigma) \,. \]    
\end{definition}
Just like in the case of the $\channels$ norm, $D_{\channels}(\varrho\|\varsigma)$ is jointly convex in $\varrho$ and $\varsigma$, because $D(\varrho\|\varsigma)$ is jointly convex.
Notice also that because $D(\varrho\|\varsigma)$ is continuous, 
then $D_{\channels}(\varrho\|\varsigma)$ is lower semi-continuous.

We are forced to introduce such generalizations, because in \Cref{sec:repeater} it will be possible to prove upper bounds on the key repeater rate in terms of some $D_\channels(\cdot\|\states)$ but not in terms of the corresponding $D_\measurements(\cdot\|\states)$.

\subsection{Bipartite systems and local norms}  
\label{sec:bipartite}

In the case where $\system{H}=\system{CD}\equiv\system{C}\otimes\system{D}$ is a tensor product Hilbert space, in other words a bipartite quantum system, two important subsets of $\allstates$ are the set of separable states $\sepstates$, and the set of PPT states $\mathcal{P}$ (positive under partial transposition), both across the bipartite cut $\system{C}{:}\system{D}$. These are defined as:
\begin{align*}
\sepstates(\system{C}{:}\system{D}) 
& \coloneqq  \conv 
\quantity{ \varrho_{\system{C}}\otimes\sigma_{\system{D}}^{\vphantom{\Gamma}}: \varrho_{\system{C}}\in\allstates(\system{C}),\ \sigma_{\system{D}}\in\allstates(\system{D}) } \, ,
\\
\mathcal{P}(\system{C}{:}\system{D}) 
& \coloneqq  
\quantity{ \varrho_{\system{CD}}\in\allstates(\system{CD}): \varrho_{\system{CD}}^{\Gamma}\in\allstates(\system{CD}) }
\\&\,= \allstates(\system{CD}) \cap \allstates(\system{CD})^\Gamma\, ,
\end{align*}
where ${(\cdot)}^\Gamma$ denotes the partial transposition (i.e.~the identity on $\system{C}$ and the transposition on $\system{D}$),
and where we define by extension $S^\Gamma \coloneqq \quantity{X^\Gamma: X\in S}$ for any set $S$ of operators on $\system{C}\system{D}$. 
The well known relative entropy of entanglement is defined as $D(\rho\|\sepstates(\system{C}{:}\system{D}))$, and is always upper bounded by $\log\min(|\system{C}|,|\system{D}|)$~\cite{VPRK}. 
The following are various important sets of channels, and corresponding sets of measurements, which capture different aspects of the subsystem separation in a bipartite system. 
Let us start with properly defining the three classes of channels that we will be interested in:
\begin{itemize}
    \item The LOCC operations $\locc(\quantum{C}{:}\quantum{D})$: These are the channels that can in principle be realized by two separated parties. More precisely they are those channels that can be achieved using local quantum channels (local operations) and classical communication. Underlining the systems indicates that the output may be a quantum system. Note that the dimension of the output system may have changed. \\
    \vspace{-1em}
    \item The separable operations $\sep(\quantum{C}{:}\quantum{D})$: These are the channels that always map separable states to separable states, even when $\system{C}$ and $\system{D}$ are entangled independently with some ancillas.
    Formally, $\Lambda_{\system{CD}}$ is a separable operation if 
    \[\id_{\system{C}'\system{D}'}\otimes\Lambda_{\system{CD}}(\sepstates(\system{C}'\system{C}{:}\system{D}'\system{D})) \subseteq\sepstates(\system{C}'\system{\tilde{C}}{:}\system{D}'\system{\tilde{D}}),\]
    where we emphasize that the output systems of $\Lambda$ may have changed dimension.

    \item The PPT operations $\mathbf{PPT(\quantum{C}{:}\quantum{D})}$: These are the channels that always map PPT states to PPT states, again even when allowing $\system{C}$ and $\system{D}$ to be entangled independently with some ancillas.
    Formally, $\Lambda_{\system{CD}}$ is a PPT operation if 
    \[\id_{\system{C}'\system{D}'}\otimes\Lambda_{\system{CD}}(\pptstates(\system{C}'\system{C}{:}\system{D}'\system{D})) \subseteq\pptstates(\system{C}'\system{\tilde{C}}{:}\system{D}'\system{\tilde{D}}),\]
    where we emphasize that the output systems of $\Lambda$ may have changed dimension.

\end{itemize}
Rigorous definitions can be found in~\cite{rains1999rigorous}.
Separable and PPT operations are not to be confused with separable and PPT channels, also known as entanglement-breaking and PPT-inducing channels. Separable channels are the channels for which $\id_{\system{C}}\otimes\Lambda_{\system{D}}(\allstates(\system{CD}))\subseteq\sepstates(\system{C}{:}\system{D})$ and PPT channels are the ones for which $\id_{\system{C}}\otimes\Lambda_{\system{D}}(\allstates(\system{CD}))\subseteq\pptstates(\system{C}{:}\system{D})$.
We have the following inclusion relations between the sets of operations above:
\[ \locc(\quantum{C}{:}\quantum{D}) \subset \sep(\quantum{C}{:}\quantum{D}) \subset \ppt(\quantum{C}{:}\quantum{D})\, .\]
For each set of operations, we can also restrict to all the measurements that can be achieved within the set. 
This defines the set of LOCC measurements
$\locc(\classic{C}{:}\classic{D})\subset\locc(\quantum{C}{:}\quantum{D})$, 
the set of separable measurements
$\sep(\classic{C}{:}\classic{D})\subset\sep(\quantum{C}{:}\quantum{D})$ 
and the set of PPT measurements
$\ppt(\classic{C}{:}\classic{D})\subset\ppt(\quantum{C}{:}\quantum{D})$, 
which also satisfy
\[ \locc(\classic{C}{:}\classic{D}) \subset \sep(\classic{C}{:}\classic{D})  \subset \ppt(\classic{C}{:}\classic{D})\, .\]
Put in a simpler way, these are all the measurements that can be obtained by composing a local measurement with an operation in the set. We make this formal below.

Define the composition of two sets of channels $\channels,\channels'$ as
\[\channels\circ\channels'\coloneqq\quantity{\Pi:\exists\ \Lambda\in\channels,\ \exists\ \Lambda'\in\channels' \text{ s.t. } \Pi=\Lambda\circ\Lambda'},\]
and define the tensor product of two sets of channels as 
\[\channels\otimes\channels'\coloneqq\quantity{\Lambda\otimes\Lambda':\Lambda\in\channels,\ \Lambda'\in\channels'} \, .\] 
The local measurements on $\system{C}\system{D}$ are then by definition  $\allmeasurements(\classic{C})\otimes\allmeasurements(\classic{D})$.

\begin{remark}
For $\channels=\loccmeasurements,\sep,\ppt$ we can rewrite the sets of measurements as 
\begin{align}
\channels(\classic{C}{:}\classic{D}) &= \quantity\big(\allmeasurements(\classic{C})\otimes\allmeasurements(\classic{D})) \circ \channels(\quantum{C}{:}\quantum{D})\, .
\end{align}
\end{remark}
\begin{proof}
We have
$\quantity\big(\allmeasurements(\classic{C})\otimes\allmeasurements(\classic{D})) \circ \channels(\classic{C}{:}\classic{D}) = \channels(\classic{C}{:}\classic{D})$, because there always exist a non-disturbing local measurement on the measurement outcomes.
Then the inclusion $\channels(\classic{C}{:}\classic{D}) \subseteq \quantity\big(\allmeasurements(\classic{C})\otimes\allmeasurements(\classic{D})) \circ \channels(\quantum{C}{:}\quantum{D})$ becomes trivial, while the opposite inclusion is trivial by definition.
\end{proof}
\noindent From now on we will use the above
as a canonical way of defining measurements.
The generalization to ``partial measurements'' then becomes straightforward (without loss of generality let $\system{C}$ be the system being measured).
\begin{definition}
For any set of channels $\channels(\quantum{C}{:}\quantum{D})$ we define
the measurement set 
\begin{align}
\label{eq:channels-to-measurements}
\channels(\classic{C}{:}\classic{D}) &\coloneqq \quantity\big(\allmeasurements(\classic{C})\otimes\allmeasurements(\classic{D})) \circ \channels(\quantum{C}{:}\quantum{D})\, ,
\intertext{and the partial measurement set}
\label{eq:channels-to-partialmeasurements}
\channels(\classic{C}{:}\quantum{D}) &\coloneqq \quantity(\allmeasurements(\classic{C}) \otimes \id_\system{D} )\circ \channels(\quantum{C}{:}\quantum{D})\, .
\end{align}
\end{definition}
Note that these definitions also apply to partial measurements on $\system{D}$ and straightforwardly generalize to multipartite systems. 
In the following sections, we will need partial measurement only in LOCC and  separable operations.

\smallskip
One of the main reasons why $\sep$ and $\ppt$
play a crucial role in quantum information theory is because they are more tractable relaxations of $\locc$. 
The sets $\sep$ and $\ppt$ 
can be characterized as in \Cref{eq:measurements-cone}, because they are generated by the cones of $\sepstates$ and $\pptstates$, respectively.
Then by \Cref{eq:M-convexcone}, the associated convex bodies $K_{\sep}$ and $K_{\ppt}$ are
\begin{align*}
K_{\sep} &= \quantity{\R^+\sepstates - \openone} \cap \quantity{\openone - \R^+\sepstates}\,,\\
K_{\ppt} &= \quantity{\R^+\pptstates - \openone} \cap \quantity{\openone - \R^+\pptstates}
= B_{\infty}\cap B_{\infty}^{\Gamma}\,.     
\end{align*}
where $B_\infty$ is the ball of the $\infty$-norm
as introduced in \Cref{sec:Mnorms}.
For $\measurements$ being $\ppt$, $\sep$ or $\locc$ we will generally refer to $\norm{\cdot}_\measurements$ as being a ``local norm''. 
Observe that in these cases we clearly have,
\[ K_{\measurements(\classic{C}{:}\classic{D})}\otimes K_{\measurements(\classic{A}'{:}\classic{B}')}\subset K_{\measurements(\classic{C}\classic{A}'{:}\classic{D}\classic{B}')}. \] 
Consequently, we find that the local norms are super-additive, namely in any dimension tensoring increases the local norm. More precisely, for any Hermitian operators $X$ on ${\system{A}'\system{B}'}$ and $Y$ on ${\system{CD}}$ we have:
\begin{equation}
\label{fact:norm-lower}
\|Y\|_{\measurements(\classic{C} {:}\classic{D} )}
\|X\|_{\measurements(\classic{A}'{:}\classic{B}')} 
\leq 
\|Y\otimes X\|
_{\measurements(\classic{C}\classic{A}'{:}\classic{D}\classic{B}')} 
\, .
\end{equation}

\begin{fact}
    \label{fact:norm-invariant}
    Let $(\measurements,\mathcal{K})$ be either $(\ppt,\mathcal{P})$ or $(\sep,\sepstates)$.
    For any Hermitian operator $X$ on ${\system{A}'\system{B}'}$ and any state $\varrho\in\allstates({\system{CD}})$
    \[ 
    \|X\|_{\measurements(\classic{A}'{:}\classic{B}')} 
    \leq 
    \|\varrho\otimes X\|
    _{\measurements(\classic{C}\classic{A}'{:}\classic{D}\classic{B}')} 
    \, . 
    \]
    If $\varrho\in\mathcal{K}(\system{C}{:}\system{D})$ then:
    \[ 
    \|X\|_{\measurements(\classic{A}'{:}\classic{B}')} 
    = 
    \|\varrho\otimes X\|
    _{\measurements(\classic{C}\classic{A}'{:}\classic{D}\classic{B}')} 
    \, .\]
\end{fact}
\begin{proof}
    The first inequality follows from \Cref{fact:norm-lower}, by simply noticing that
    $\|\varrho\|_{\measurements}=\|\varrho\|_1 = 1$ on positive operators 
    (because the identity is contained in $K_{\mathbf M}$ and $\Tr(\cdot)=\|\cdot\|_1$ on positive operators).
    \\For the opposite inequality in the case of  $\varrho\in\mathcal{K}(\system{C}{:}\system{D})$,
    we have that the state preparation channel $\Lambda(X) = \varrho \otimes X$ is in $\measurements(\quantum{A'}{:}\quantum{B'})$,
    and therefore for any Hermitian operator $X$ on $\system{A}'\system{B}'$,
    \begin{align*} 
    \|\varrho\otimes X\|_{\measurements(\classic{C}\classic{A}'{:}\classic{D}\classic{B}')} 
      &= \|\Lambda(X)\| _{\measurements(\classic{A}'{:}\classic{B}')}
       \leq \|X\| _{\measurements(\classic{A}'{:}\classic{B}')}
    \end{align*}
    which proves the claim.
\end{proof}

\subsection{Continuity bounds on relative entropy measures}

In the following we will generalise the asymptotic continuity bound for the measured relative entropy measures \cite{LiW} to measurements on only part of the system.

\begin{proposition}[Asymptotic continuity of the $\channels$ relative entropy]
    \label{prop:D_L|K0} Let $\states$ be a set of states star-shaped around the maximally mixed state and let $\mathbf{L}  :\,= \id_{\rA}\otimes\mathbf{ALL}(\rB)$.
    Then, for any states $\varrho$ and $\varrho'$ on $AB$ satisfying $\epsilon \coloneqq \left\|\varrho-\varrho'\right\|_\channels/2$, we have
    \[ \quantity|D_\channels(\varrho\|\states) - D_\channels(\varrho'\|\states)| \leq \kappa \epsilon  \log d + g(\epsilon)\,,\]
where $d=\dim A= \dim B$, $\kappa=16$ and $g(\epsilon) =8 \epsilon \log 3 + 2\nu(2\epsilon) + h(2\epsilon)$ with $h(\cdot)$ the binary entropy function and $\nu(t)=-t\log t$.
\end{proposition}

\begin{proof}

$\mathcal{K} \equiv \mathcal{K}(\rA\rB) \subset \mathcal{D}(\rA\rB)$ containing $\tau$ and star-shaped w.r.t.~$\tau$, the maximally mixed state. Set $\mathcal{K}_x :\,= (1-x)\tau+x\mathcal{K}$ with $x$ to be chosen later.

We fix $\Lambda\in\mathbf{L}$, i.e.~$\Lambda_{\rA\rB}=\id_{\rA}\otimes\mathcal{M}_{\rB} $, with $\mathcal{M}=(M^i_{\rB})_i$ a measurement on $\rB$. It will be useful for us later on to re-write, for each $i$, $M^i_{\rB}=3d\lambda_iQ^i_{\rB}$, where $\Tr Q^i_{\rB}=1/3$. Then, for any $\rho\in\mathcal{D}$, we have
\[  \Lambda(\rho)= \sum_i \mu_i{\rho_i}_{\rA}\otimes\ketbra{i}{i}_{\rB} = \sum_i 3d \lambda_i \hat{\mu}_i{\rho_i}_{\rA}\otimes\ketbra{i}{i}_{\rB} \, , \]
where $\mu_i:\,=\Tr(\rho_{\rB}\,M^i_{\rB})$, $\hat{\mu}_i:\,=\Tr(\rho_{\rB}\,Q^i_{\rB})$ and ${\rho_i}_{\rA}:\,=\Tr_{\rB}(\rho_{\rA\rB}\,\openone_{\rA}\otimes M^i_{\rB})/\Tr(\rho_{\rB}\,M^i_{\rB})$.
We will also need to define
\begin{align*}
& \epsilon_{\mathcal{M}} :\,= \frac{1}{2} \|\mathcal{M}(\rho_{\rB})-\mathcal{M}(\rho'_{\rB})\|_1 = \frac{1}{2} \sum_i |\mu_i-\mu'_i| \, , \\
& \epsilon_{\Lambda} :\,= \frac{1}{2} \|\Lambda(\rho)-\Lambda(\rho')\|_1 = \frac{1}{2} \sum_i \|\mu_i\rho_i-\mu_i'\rho_i'\|_1 \, .
\end{align*}

We now fix $\sigma\in\mathcal{K}_x$. Then, for any $\rho\in\mathcal{D}$, we have
\begin{align*} 
D(\Lambda(\rho)\|\Lambda(\sigma)) & = \sum_i \Tr(\mu_i\rho_i(\log(\mu_i\rho_i)-\log(\eta_i\sigma_i))) \\
& = \sum_i \Tr(\mu_i\rho_i(\log(\hat{\mu}_i\rho_i)-\log(\hat{\eta}_i\sigma_i))) \, ,
\end{align*}
where $\eta_i:\,=\Tr(\sigma_{\rB}\,M^i_{\rB})$, $\hat{\eta}_i:\,=\Tr(\sigma_{\rB}\,Q^i_{\rB})$ and ${\sigma_i}_{\rA}:\,=\Tr_{\rB}(\sigma_{\rA\rB}\,\openone_{\rA}\otimes M^i_{\rB})/\Tr(\sigma_{\rB}\,M^i_{\rB})$.
And hence, for any $\rho,\rho'\in\mathcal{D}$, we have
\begin{align*} 
& | D(\Lambda(\rho)\|\Lambda(\sigma)) - D(\Lambda(\rho')\|\Lambda(\sigma)) | \\
& \ \ \leq \sum_i | \Tr(\mu_i\rho_i\log(\hat{\mu}_i\rho_i)) - \Tr(\mu_i'\rho_i'\log(\hat{\mu}_i'\rho_i')) | \\ 
& \ \ \ \ + \sum_i | \Tr(\mu_i\rho_i\log(\hat{\eta}_i\sigma_i)) - \Tr(\mu_i'\rho_i'\log(\hat{\eta}_i\sigma_i)) | \, .
\end{align*}

Let us start with upper bounding the first term. Observe that
\begin{align*} 
& |\Tr(\mu_i\rho_i\log(\hat{\mu}_i\rho_i)) - \Tr(\mu_i'\rho_i'\log(\hat{\mu}_i'\rho_i')) | \\
& \ \ \leq |\mu_i\log\hat{\mu}_i-\mu_i'\log\hat{\mu}_i'| + |\mu_i S(\rho_i) - \mu_i' S(\rho_i')| \, , 
\end{align*}
where $S(\cdot )$ denotes the von Neumann entropy. Now, on the one hand, setting $\nu(t)=-t\log t$, we have
\begin{align*}
\sum_i |\mu_i\log\hat{\mu}_i-\mu_i'\log\hat{\mu}_i'| & = 3d\sum_i \lambda_i |\nu(\hat{\mu}_i) - \nu(\hat{\mu}_i') | \\
& \leq 3d \sum_i \lambda_i \nu(|\hat{\mu}_i-\hat{\mu}_i'|) \\
& \leq 3d \nu\left( \sum_i \lambda_i |\hat{\mu}_i-\hat{\mu}_i'| \right) \\
& = 3d \nu\left( \frac{\|\mathcal{M}(\rho_{\rB})-\mathcal{M}(\rho'_{\rB})\|_1}{3d} \right) \, ,
\end{align*}
where the first inequality is because $\nu(t+s)\leq\nu(t)+\nu(s)$ and the second inequality is because $\nu$ is concave. And we thus have shown that
\[ \sum_i |\mu_i\log\hat{\mu}_i-\mu_i'\log\hat{\mu}_i'| \leq 2\epsilon_{\mathcal{M}}\log(3d) +\nu(2\epsilon_{\mathcal{M}}) \, . \]
Then, on the other hand we know that, for any $0\leq\mu,\mu'\leq 1$ and $\varsigma_{\rA},\varsigma'_{\rA}\in\mathcal{D}$, we have
\begin{align*}
& \left| \mu S(\varsigma_{\rA}) - \mu' S(\varsigma'_{\rA}) \right| \\
& \ \ = \left| \mu S(\varsigma_{\rA}) - \mu'S(\varsigma_{\rA}) +  \mu'S(\varsigma_{\rA}) - \mu'S(\varsigma'_{\rA}) \right| \\
& \ \ \leq \left| \mu-\mu' \right| S(\varsigma_{\rA}) + \mu' \left| S(\varsigma_{\rA})-S(\varsigma'_{\rA}) \right| \\
& \ \ \leq \left| \mu-\mu' \right|\log d \\
& \ \ \ \ + \mu' \left( \frac{\|\varsigma_{\rA}-\varsigma_{\rA}'\|_1}{2}\log d + h\left( \frac{\|\varsigma_{\rA}-\varsigma_{\rA}'\|_1}{2} \right) \right) \, ,
\end{align*}
where $h(t)=-t\log t -(1-t)\log(1-t)$. The last inequality follows from Fannes inequality (\cite{Fannes,Audenaert,Petz}) in the form of \cite[Lemma 1]{Winter}. Now,
\begin{align*}
\mu' \|\varsigma_{\rA}-\varsigma_{\rA}'\|_1  & = \| \mu'\varsigma_{\rA} - \mu\varsigma_{\rA} + \mu\varsigma_{\rA} - \mu'\varsigma_{\rA}' \|_1 \\
& \leq |\mu'-\mu| + \| \mu\varsigma_{\rA} - \mu'\varsigma_{\rA}' \|_1 \, .
\end{align*} 
Therefore,
\begin{align*}
\sum_i & | \mu_i S(\rho_i) - \mu_i' S(\rho_i') | \\ 
 &  \leq \frac{3}{2}\sum_i | \mu_i- \mu_i'|\log d + \frac{1}{2}\sum_i\| \mu_i \rho_i - \mu'_i \rho'_i \|_1  \log d \\
 & \ \ \ \ + \sum_i \mu_i' h\left( \frac{\| \rho_i - \rho'_i \|_1}{2} \right) \, .
\end{align*}
And by concavity of $h$,
\begin{align*}
\sum_i & \mu_i' h\left( \frac{\| \rho_i - \rho'_i \|_1}{2} \right) \\
& \leq h\left( \sum_i \mu_i' \frac{\| \rho_i - \rho'_i \|_1}{2} \right) \\
& \leq h\left( \sum_i \frac{|\mu_i-\mu_i'|}{2} + \sum_i \frac{\| \mu_i \rho_i - \mu_i' \rho'_i \|_1}{2} \right) \, ,
\end{align*}
where the last inequality holds for $\epsilon_{\mathcal{M}} + \epsilon_{\Lambda} \leq 1/2$ since $h$ is non-decreasing on $[0,1/2]$. We thus have
\begin{align*} 
\sum_i & | \mu_i S(\rho_i) - \mu_i' S(\rho_i') | \\
& \leq 3\epsilon_{\mathcal{M}}\log d + \epsilon_{\Lambda}\log d + h(\epsilon_{\mathcal{M}} +\epsilon_{\Lambda}) \, .
\end{align*}
Hence, putting everything together, we eventually get
\begin{align*} 
\sum_i &| \Tr(\mu_i\rho_i\log(\hat{\mu}_i\rho_i)) - \Tr(\mu_i'\rho_i'\log(\hat{\mu}_i'\rho_i')) |\\ 
& \leq 2\epsilon_{\mathcal{M}}\log(3d) +\nu(2\epsilon_{\mathcal{M}}) \\
& \ \ \ \ + 3\epsilon_{\mathcal{M}}\log d + \epsilon_{\Lambda}\log d + h(\epsilon_{\mathcal{M}} +\epsilon_{\Lambda}) \\ 
& \leq 6\epsilon\log(3d)  + \nu(2\epsilon) + h(2\epsilon) \, ,
\end{align*}
where the last inequality holds for $\epsilon\leq 1/(2e)$, since $\epsilon_{\mathcal{M}} \leq \epsilon_{\Lambda} \leq \epsilon$ and $\nu$, resp.~$h$, is non-decreasing on $[0,1/e]$, resp.~$[0,1/2]$.

Let us now turn to upper bounding the second term. 
\begin{align*}  
\sum_i & | \Tr((\mu_i\rho_i-\mu_i'\rho_i')\log(\hat{\eta}_i\sigma_i)) |\\
& \leq \sum_i \|\mu_i\rho_i-\mu_i'\rho_i'\|_1 \,\underset{i}{\max} \|\log(\hat{\eta}_i\sigma_i)\|_{\infty} \\ 
& \leq 2\epsilon_{\Lambda}\log\left(\frac{3d^2}{1-x}\right) \, ,
\end{align*}
where the last inequality is because $\|\log(\hat{\eta}_i\sigma_i)\|_{\infty} \leq \log(3d^2/(1-x))$ arising from a lower bound on the minimal eigenvalue of the argument of the logarithm (see definition of $x$ at beginning of the proof). Thus, choosing $x=1-2\epsilon$, we finally obtain
\[ \sum_i | \Tr((\mu_i\rho_i-\mu_i'\rho_i')\log(\hat{\eta}_i\sigma_i)) | \leq 2\epsilon\log(3d^2) +\nu(2\epsilon) \, , \]
again for $\epsilon\leq 1/(2e)$.

Consequently, we have proved that, for any $\Lambda\in\mathbf{L}$ and $\sigma\in\mathcal{K}_{1-2\epsilon}$, we have
\begin{align*} 
& |  D(\Lambda(\rho)\|\Lambda(\sigma)) - D(\Lambda(\rho')\|\Lambda(\sigma)) | \\
& \ \ \leq 8\epsilon\log(3d^2)  + 2\nu(2\epsilon) + h(2\epsilon) \, . 
\end{align*}

\end{proof}

It is easy to generalise the statement to the case, where $A$ and $B$ have different dimension, but we will omit this generalisation, since we do not need it in the following. 

We believe this statement is also true when the the state can be preprocessed by LOCC (even if the LOCC is enlarging the initial dimensions of the system). 

\begin{conjecture}[Asymptotic continuity of the $\channels$ relative entropy]
    \label{prop:D_L|K}
    Let $\states$ be a set of states star-shaped around the maximally mixed state and let
$\mathbf{L} \equiv \mathbf{L}(\underline{\rA}{:}\rB) :\,= (\id_{\rA}\otimes\mathbf{ALL}(\rB)) \circ \mathbf{L}(\underline{\rA}{:}\underline{\rB})$.
    Then, for any states $\varrho$ and $\varsigma$ on $AB$ satisfying $\epsilon \coloneqq \left\|\varrho-\varsigma\right\|_\channels/2$, we have
    \[ \quantity|D_\channels(\varrho\|\states) - D_\channels(\varsigma\|\states)| \leq \kappa \epsilon  \log d + g(\epsilon)\,,\]
where $d=\dim (\system{H}_A) =\dim (\system{H}_B)$, for some constant $\kappa$ and a continuous function $g(\epsilon)$ satisfying $g(\epsilon) \rightarrow 0$ for $\epsilon \mapsto 0$.
\end{conjecture}

In fact, we might envision that the statement even holds true for any set of quantum channels. Our belief is based on the fact that the statement holds, on the one hand, for any set of measurements, and, on the other hand, for the set of channels consisting only of the identity channel. In a certain sense these two sets of channels are on the extremes of an arbitrary set of channels and thus a continuity statement could be expected to be true there, too. With regards to the stated explicit conjecture which concerns a very specific class of intermediate classes of channels we point out that our Proposition \ref{prop:D_L|K0} is a natural special case of it.

\subsection{Local norm increase under tensoring} 
\label{sec:norm-increase}

We have just seen from \Cref{fact:norm-invariant} that not all states increase a local norm.
Indeed, PPT states do not increase the PPT norm and similarly, separable states do not increase the SEP norm.
In this section we study how general entangled states can increase the local norms. 
Namely, we study how tensoring a Hermitian operator $X$ on ${\system{A}'\system{B}'}$ (which for our purposes can be thought of as a difference of two states) with a state $\varrho$ on ${\system{CD}}$ changes the $\measurements$ norm, for $\measurements$ being either $\ppt$ or $\sep$.
The question we are now interested in is to get an upper bound on $\|\varrho\otimes X\|_{\measurements(\classic{A}\classic{A}'{:}\classic{B}\classic{B}')}$ in terms of $\|X\|_{\measurements(\classic{A}'{:}\classic{B}')}$. 
\smallskip

In the next section, we will be interested in proving that some pairs of states have a large PPT norm distance and a small SEP norm distance. We will thus focus on finding lower bounds, and not upper bounds, on the PPT norm. However, we will state the upper bounds also on the PPT norm for the sake of completeness.
For our statement we need the \emph{robustness of entanglement} \cite{TV}, 
which for any state $\varrho$ on $\system{C}\system{D}$ is defined as
\begin{align}
\label{eq:robustness}
\robustness(\varrho) & \coloneqq 
\inf_{\sigma\in\sepstates(\system{C}{:}\system{D})} \robustness(\varrho\|\sigma).
\end{align}
where 
\begin{align*}
\robustness(\varrho\|\sigma) &\coloneqq \inf\quantity{s: \frac{1}{1+s}\quantity(\varrho+s\,\sigma)\in\sepstates(\system{C}{:}\system{D})}.
\end{align*}

This is not to be confused with the \emph{global robustness of entanglement}
in which $\sigma$ is allowed to vary over all states in $\allstates(\system{CD})$~\cite{datta2009max}.

\begin{proposition}
\label{prop:psi-M}
For any Hermitian operator $X$ on $\system{A}\system{B}$ and any state $\varrho$ on $\system{C}\system{D}$, we have
\begin{align*}
& \|\varrho\otimes X\|
_{\sep(\classic{C}\classic{A}{:}\classic{D}\classic{B})} 
\leq (2R(\varrho)+1)\|X\|_{\sep(\classic{A}{:}\classic{B})}, \\
& \|\varrho\otimes X\|
_{\ppt(\classic{C}\classic{A}{:}\classic{D}\classic{B})} 
\leq \|\varrho^{\Gamma}\|_1\|X\|_{\ppt(\classic{A}{:}\classic{B})}. 
\end{align*}
Setting $k=\min(|\system{C}|,|\system{D}|)$, we therefore have
\begin{align*}
& \|\varrho\otimes X\|
_{\sep(\classic{C}\classic{A}{:}\classic{D}\classic{B})} \leq (2k-1)\|X\|_{\sep(\classic{A}{:}\classic{B})}, \\
& \|\varrho\otimes X\|
_{\ppt(\classic{C}\classic{A}{:}\classic{D}\classic{B})} \leq k\,\|X\|_{\ppt(\classic{A}{:}\classic{B})}. 
\end{align*}
The third equality is a direct consequence of the proof of~\cite[Theorem 16]{lami2017ultimate}.
\end{proposition}

\begin{proof}
The second set of inequalities in the proposition is easily derived from the first one, after upper bounding the maximal value that $R(\varrho)$ and $\|\varrho^{\Gamma}\|_1$ might take. The fact that $\|\varrho^{\Gamma}\|_1\leq k$ is well-known. While it was shown in \cite[Theorem C.2]{TV} that $R(\varrho)\leq k-1$.

For the SEP norm, we follow an argument inspired by~\cite[Theorem 16]{lami2017ultimate}.
We know that there exists a separable state $\tau$ which is such that, setting $R\coloneqq \robustness(\varrho)$, the following state is also separable: 
\[ \varrho' = \frac{1}{1+R}\varrho + \frac{R}{1+R}\tau. \] 
Because the SEP norm is left unchanged under tensoring with a separable state, we have
\begin{align*}
\norm{X}_\sep 
  &= \norm{\varrho' \otimes X}_\sep 
\\&\geq \frac{1}{1+R}\,\norm{\varrho \otimes X}_\sep 
   - \frac{R}{1+R}\,\norm{\tau \otimes X} _\sep
\\&= \frac{1}{1+R}\,\norm{\varrho \otimes X} _\sep
   - \frac{R}{1+R}\,\norm{X}_\sep, 
\end{align*}
where we used the triangle inequality and the separability of $\varrho'$ and $\varrho$.
Hence, we obtain as announced that 
$(2R+1)\norm{X}_\sep\geq \norm{\varrho \otimes X}_\sep$.

For the PPT norm notice first that, because $K_{\ppt}=B_{\infty}\cap B_{\infty}^{\Gamma}$, its polar is simply 
\begin{align*}
K_{\ppt}^{\circ}&=\conv\quantity(B_1\cup B_1^{\Gamma})
\\&=\quantity{ \lambda Y+(1-\lambda)Z:\ \norm{Y}_1,\norm*{Z^{\Gamma}}_1\leq 1,\ \lambda\in [0,1]},
\end{align*} 
and therefore we have
\begin{align*}
\|\varrho\otimes X\|_{\ppt} 
&=\inf\quantity{\mu: \varrho\otimes X \in \mu \conv\quantity(B_1\cup B_1^{\Gamma}) }.
\intertext{Since we restrict to finite dimensions, 
the minimisation extends over a compact set and has thus a minimizer. We therefore obtain}
\|\varrho\otimes X\|_{\ppt} 
&=\min\quantity{\mu: \varrho\otimes X \in \mu \conv\quantity(B_1\cup B_1^{\Gamma}) }.
\\&=\min\big\{\mu: \varrho\otimes X =\lambda Y+(1-\lambda)Z,
\\&\qquad\quad \norm{Y}_1\leq \mu,\ \norm*{Z^{\Gamma}}_1\leq \mu,\ \lambda\in[0,1]\big\}
\\
&=\min\big\{\max(\norm{Y}_1,\norm*{Z^{\Gamma}}_1): 
\\&\qquad\quad\varrho\otimes X =\lambda Y+(1-\lambda)Z,\ \lambda\in[0,1] \big\}.
\intertext{Now, let $X=\lambda_0Y_0+(1-\lambda_0)Z_0$ such that $\|X\|_{\ppt}=\max \quantity( \|Y_0\|_1, \|Z_0^{\Gamma}\|_1)$ as just derived. Since $\varrho\otimes X=\lambda_0\varrho\otimes Y_0+(1-\lambda_0)\varrho\otimes Z_0$, we then have}
\|\varrho\otimes X\|_{\ppt} 
&\leq \max\quantity( \|\varrho\otimes Y_0\|_1, \|{(\varrho\otimes Z_0)}^{\Gamma}\|_1 ) \\
&= \max \quantity( \|\varrho\|_1\|Y_0\|_1, \norm*{\varrho^\Gamma}_1 \|Z_0^{\Gamma}\|_1 ) \\
&\leq \norm*{\varrho^\Gamma}_1 \max \left( \|Y_0\|_1, \|Z_0^{\Gamma}\|_1 \right) \\ 
&= \norm*{\varrho^\Gamma}_1 \|X\|_{\ppt},
\end{align*}
the first equality being by multiplicativity of the trace norm under tensoring and the second inequality being because $\norm*{\varrho^\Gamma}_1\geq \|\varrho\|_1$.
\end{proof}

\begin{remark}
The case we will in particular focus on in the remainder of this paper is when the considered state $\varrho$ on ${\system{CD}}$ is the maximally entangled state $\psi\coloneqq\sum_{i,j=1}^k \ketbra{ii}{jj}/k$, for which $R(\psi)=k-1$ and $\|\psi^{\Gamma}\|_1=k$. 
\end{remark}

A legitimate question at this point is that of optimality in \Cref{prop:psi-M}. Indeed the lower bound in \cite[Proposition 16]{lami2017ultimate} gives the following statement (the construction given for $X$ is a weighted difference of the symmetric and antisymmetric projectors), which shows that the maximally entangled state can achieve an almost optimal increase in local norm.

\begin{proposition} \cite[Proposition 16]{lami2017ultimate}
\label{prop:optimality}
Let $|\system{C}|=|\system{D}|=k$ and let $\psi$ be the maximally entangled state on $\system{C}\system{D}$ as above. For any $k$
there exists a Hermitian operator $X$ on ${\system{A}\system{B}}$ such that for $\measurements$ being either $\sep$ or $\ppt$ it holds
\begin{align*}
 \left\|\psi\otimes X\right\|
 _{\measurements(\classic{C}\classic{A}{:}\classic{D}\classic{B})} 
 \geq k\left\|X\right\|_{\measurements(\classic{A}{:}\classic{B})}
\end{align*}
\end{proposition} 

We will apply the results of this section in the following to the case where $X$ is a difference of two states $\rho^+$ and $\rho^-$. We then see that our bounds give limitations on the power of distinguishing the two states when separable or PPT operations are assisted by entangled states of certain fixed dimension. This will come in handy when discussing private state, as they are built from $\rho^\pm$ as well as Bell states.

\subsection{Application to private quantum states}

From now on, when we talk about the sets of states $\pptstates$ or $\sepstates$, we might omit the system labels in the proofs for ease of reading; in such case $\pptstates$ and $\sepstates$ are always assumed to be according to a bipartite cut (for example $\system{A}{:}\system{B}$, $\system{A}'{:}\system{B}'$ or $\system{A}\system{A}'{:}\system{B}\system{B}'$) which will be clear from the context or from the theorem statements. The same holds for $\ppt$ and $\sep$. 

\smallskip

Let us now make precise the concepts that we informally defined in the introduction. On a two-qubits system
we denote by $\psi^+ = \psi$ and $\psi^-$ the two Bell states
\begin{align*}
\psi^\pm := & \frac{1}{2}\left(\ket{00}\pm\ket{11}\right)\left(\bra{00}\pm\bra{11}\right)\, . 
\end{align*}
Measuring these states in the computational basis leads to a bit of key: 
two perfectly correlated, perfectly random bits (one at each system)
that are secret from the environment. 

Private quantum states with one bit of key are states that generalize the Bell states to any bipartite system $\system{A}\system{A}'\system{B}\system{B}'$ with $|\system{A}|=|\system{B}|=2$. They generalize the Bell states in the sense that measuring in the computational basis of $\system{A}\system{B}$ still yields a bit of key. The bit in $\system{AB}$ might be correlated with $\system{A}'\system{B}'$ but it will be secret from any purifying environment. For any pair of orthogonal states  $\shieldplusminus$ on $\system{A}'\system{B}'$, we can construct the following private state~\cite{HHHO}
\begin{equation}
\label{eq:def-gamma} 
\gamma_{\system{A}\system{A}'\system{B}\system{B}'} 
\coloneqq \frac{1}{2}\left(\psi^+_{\system{AB}}\otimes\shieldplus_{\system{A}'\system{B}'} 
+\psi^-_{\system{AB}}\otimes\shieldminus_{\system{A}'\system{B}'}\right)\, .
\end{equation}
Systems $\system{A}\system{B}$ are called key systems and systems $\system{A}'\system{B}'$ are called shield systems. 

The state obtained after the measurement in the computational basis of $\system{A}\system{B}$ is called the key-attacked state. For the private state given by \Cref{eq:def-gamma} it equals
\begin{equation}
\label{eq:gammahat} 
\hat{\gamma}_{\system{A}\system{A}'\system{B}\system{B}'} 
= \frac{1}{4}\left(\psi^+_{\system{AB}}+\psi^-_{\system{AB}}\right) 
\otimes\left(\shieldplus_{\system{A}'\system{B}'}+\shieldminus_{\system{A}'\system{B}'}\right)\, . 
\end{equation}

Define the Hermitian operator $\Delta\coloneqq(\shieldplus-\shieldminus)/2$, then $|\Delta|=(\shieldplus+\shieldminus)/2$
and in matrix notation we have 
\begin{align*} 
\gamma_{\system{A}\system{A}'\system{B}\system{B}'} = \frac{1}{2} \matrixquantity ( |\Delta_{\system{A}'\system{B}'}| & 0 & 0 & \Delta_{\system{A}'\system{B}'} \\ 0 & 0 & 0 & 0 \\ 0 & 0 & 0 & 0 \\ \Delta_{\system{A}'\system{B}'} & 0 & 0 & |\Delta_{\system{A}'\system{B}'}| ) \, \phantom.
\intertext{and}
\hat\gamma_{\system{A}\system{A}'\system{B}\system{B}'} = \frac{1}{2} \matrixquantity( |\Delta_{\system{A}'\system{B}'}| & 0 & 0 & 0 \\ 0 & 0 & 0 & 0 \\ 0 & 0 & 0 & 0 \\  0 & 0 & 0 & |\Delta_{\system{A}'\system{B}'}| )
\, .
\end{align*}

We know from~\cite{CF} that any private state can be transformed via a reversible LOCC operation into a private state of the form in \Cref{eq:def-gamma}. It is therefore legitimate to focus only on the construction of such so-called Bell private states.
\smallskip

With the results of \Cref{sec:norm-increase}, it is now not hard to see that, if $\shieldplusminus$ on ${\system{A}'\system{B}'}$ are data-hiding for PPT or SEP measurements on $\system{A}'\system{B}'$, then so is the constructed private state  $\gamma$ on ${\system{A}\system{A}'\system{B}\system{B}'}$  for PPT or SEP measurements on $\system{A}\system{A}'\system{B}\system{B}'$.

\begin{lemma} 
\label{lemma:pbit-M-norms-upper}
    Let $\gamma$ and $\hat{\gamma}$ on ${\system{A}\system{A}'\system{B}\system{B}'}$ be a private state and its key-attacked state, as defined by \Cref{eq:def-gamma,eq:gammahat},
    and let $\shieldplusminus$ on ${\system{A}'\system{B}'}$ be the corresponding shield states. Then,
    \begin{align*}
    & \|\gamma-\hat{\gamma}\|_{\sep(\classic{A}\classic{A}'{:}\classic{B}\classic{B}')} 
    \leq  \|\shieldplus-\shieldminus\|_{\sep(\classic{A}'{:}\classic{B}')}
    \, , \\
    & \|\gamma-\hat{\gamma}\|_{\ppt(\classic{A}\classic{A}'{:}\classic{B}\classic{B}')} 
    \leq          \|\shieldplus-\shieldminus\|_{\ppt(\classic{A}'{:}\classic{B}')}
    \, .
    \end{align*}
    If the key-attacked state is separable, this implies
        \begin{align*}
    &\|\gamma-\sepstates(\system{A}\system{A}'{:}\system{B}\system{B}')\|
    _{\sep(\classic{A}\classic{A}'{:}\classic{B}\classic{B}')} \leq
    \|\shieldplus-\shieldminus\|_{\sep(\classic{A}'{:}\classic{B}')}
    \, , \\
    & \|\gamma-\sepstates(\system{A}\system{A}'{:}\system{B}\system{B}')\|
    _{\ppt(\classic{A}\classic{A}'{:}\classic{B}\classic{B}')} \leq 
    \|\shieldplus-\shieldminus\|_{\ppt(\classic{A}'{:}\classic{B}')}
    \, .
    \end{align*}

\end{lemma}

\begin{proof}
    Note that the states $\gamma$ and $\hat{\gamma}$ are such that
    \[ \gamma -\hat{\gamma}
    = \frac{1}{4}\left(\psi^+-\psi^-\right)
    \otimes\left(\shieldplus-\shieldminus\right). \]
Now rewrite
$$\psi^+-\psi^-= 2 \tau^+-2\tau^-,$$
where 
$$\tau^{\pm} := 1/4 (\id \pm |00\rangle\langle 11| \pm |11\rangle\langle 00|).$$
Notice that $\tau^{\pm}$ are separable by the Horodecki PPT criterion which is sharp if the local dimension is two. 
 We therefore find  with help of the triangle inequality that 
    \begin{align*}
    \left\|\gamma - \hat\gamma\right\|_{\mathbf{M}} 
      &\leq \frac{1}{2} \norm{\tau^+\otimes\quantity(\shieldplus-\shieldminus)}_\measurements
    \\&\,+ \frac{1}{2} \norm{\tau^-\otimes\quantity(\shieldplus-\shieldminus)}_\measurements
    \\&= \norm{\shieldplus-\shieldminus}_\measurements . 
    \end{align*}
    The two announced inequalities then follow from \Cref{fact:norm-invariant}.
\end{proof}

The result for PPT measurements was provided for completeness, 
as later we will actually want to prove that our constructed private states are not data-hiding for PPT measurements. The following lemma will be useful in this context. 

\begin{lemma} 
    \label{lemma:pbit-M-norms-lower}
    Let $\gamma$ on ${\system{A}\system{A}'\system{B}\system{B}'}$ be a private state, as defined by \Cref{eq:def-gamma}, and let let $\shieldplusminus$ on ${\system{A}'\system{B}'}$ be the corresponding shields. Then
    \begin{align*}
    \|\gamma-\sepstates(\system{A}\system{A}'{:}\system{B}\system{B}')\|_{\ppt(\classic{A}\classic{A}'{:}\classic{B}\classic{B}')} 
    &\geq \frac{1}{3} \|\shieldplus-\shieldminus\|_{\ppt(\classic{A}'{:}\classic{B}')}\,.
    \end{align*}
\end{lemma}

\begin{proof}
    Let $c\coloneqq\frac{1}{2}\norm{\shieldplus-\shieldminus}_{\ppt(\classic{A}'{:}\classic{B}')} $.
    By definition, this means that there exist a binary measurement $(M,\overline M =\openone-M)\in\ppt(\classic{A}'{:}\classic{B}')$ such that without loss of generality we have
    \[\Tr M \shieldplus - \Tr M \shieldminus  = c
    = \Tr \overline M \shieldminus - \Tr \overline M \shieldplus\,.\]
    We now apply on $\gamma$ a PPT distillation protocol $\Lambda$ that first tries to distinguish $\shieldpm$ using the above measurement, and then corrects the phase flip of the maximally entangled state accordingly.
    Since after the correction the measurement outcome is not needed anymore, it is traced out at the end of the protocol.
    The resulting state is 
    \begin{align*}
    \Lambda(\gamma) 
      &= \frac{1}{2}\Tr(M\shieldplus +\overline M\shieldminus) \psi^+
       + \frac{1}{2}\Tr(M\shieldminus+\overline M\shieldplus ) \psi^-
    \\&= \frac{1}{2}(1+c) \psi^+
       + \frac{1}{2}(1-c) \psi^-\;.
    \end{align*}
    
    We now apply an isotropic twirl
    to $\system{A}\system{B}$
    to produce an isotropic state~\cite{HH}.
    The twirl produces an isotropic state $\tilde\gamma$ with fidelity to the maximally entangled state $\Tr \psi^+\tilde\gamma = (1-c)/2$ (see Appendix \ref{section:iso} for more details).
    Notice that $\tilde\gamma$ is always entangled,
    as it was proven in~\cite{MWW} that $c\geq1/|\system{A}|$,
    and the isotropic states are entangled as soon as the fidelity is bigger that $1/|\system{A}|$~\cite{HH}.
    
    Since all of the above operations are within the PPT operations, 
    they map PPT states into PPT states (see \Cref{sec:bipartite}).
    Together with the fact that all PPT isotropic states are also separable,
    we find that the PPT norm of $\gamma$ decreases. Namely, denote by $\tilde\Lambda$ the PPT operation $\Lambda$ followed by a twirl, 
    and let $\sigma\in\sepstates(\system{A}\system{A}'{:}\system{B}\system{B}')$
    be the state such that $\norm{\gamma-\sepstates}_\pptchannels=\norm{\gamma-\sigma}_\pptchannels$ then
    \begin{align*}
    \norm{\gamma-\sepstates}
    _{\pptchannels(\system{A}\system{A}'{:}\system{B}\system{B}')}
      &=
        \norm{\gamma-\sigma}
    _{\pptchannels(\system{A}\system{A}'{:}\system{B}\system{B}')}
    \\&\geq
        \norm{\tilde\Lambda(\gamma-\sigma)}
    _{\pptchannels(\system{A}{:}\system{B})}
    \\&=
        \norm{\tilde\gamma-\tilde\Lambda(\sigma)}
    _{\pptchannels(\system{A}{:}\system{B})}
    \\&\geq
        \norm{\tilde\gamma-\sepstates(\system{A}{:}\system{B})}
    _{\pptchannels(\system{A}{:}\system{B})}
    \end{align*}
    where we used that $\tilde\Lambda(\sigma)\in\sepstates(\system{A}{:}\system{B})$.
    \Cref{lemma:iso:distinguishability} now gives us the desired lower bounds for $|\system{A}|=|\system{B}|=2$:
    \begin{align*}
    \norm{\tilde\gamma-\sepstates}_\pptchannels   
    &=
        \frac{4}{3}\quantity(\frac{1+c}{2}-\frac{1}{2})
        = \frac{2}{3}c
    \qedhere
    \end{align*}
\end{proof}

\section{Random private quantum states}
\label{sec:p-bit}

\subsection{Random private quantum state construction}

With \Cref{lemma:pbit-M-norms-upper,lemma:pbit-M-norms-lower} in mind, we now turn to the objective of generating random private quantum states with interesting properties.
Our construction of random private quantum states will be based on a construction of random orthogonal quantum states which was introduced in~\cite[Section 6.1]{AL}, and which we recall here.
Notice that from now on, $\system{A}'$ and $\system{B}'$ will be fixed to be $d$-dimensional complex Hilbert spaces for some $d\in\N$.

\begin{construction}[Random orthogonal quantum states]
    \label{con:random-states}\label{eq:def-rho-sigma'} 
    Let $|\system{A}'|=|\system{B}'|=d$ and without loss of generality assume that $d$ is even, and let $P$ be an orthogonal projector on some 
    fixed $d^2/2$-dimensional subspace of $\system{A}'\system{B}'$. 
    Define first the two following orthogonal states: 
    \[ 
    \fixedshieldplus \coloneqq \frac{P}{\Tr P}
    \ \ \text{and}\ \ 
    \fixedshieldminus \coloneqq \frac{P^{\perp}}{\Tr P^{\perp}}
    \, . \]
    Then, let $U$ be a Haar-distributed random unitary on $\system{A}'\system{B}'$, 
    and define the two following random orthogonal states: 
    \[ \randomshieldplusminus\coloneqq 
    U\fixedshieldplusminus U^\dagger. \]
    \end{construction}

\begin{lemma}
\label{lemma:robustness:random}
Let $\randomshield$ be a random state on $\system{A}'\system{B}'$ as in \Cref{eq:def-rho-sigma'} and let $\tau=\openone/d^2$ be the maximally mixed state on $\system{A}'\system{B}'$. Namely, let $\randomshield = U P U^\dagger / (d^2/2)$ where $P$ is the projector on a $d^2/2$-dimensional subspace of $\system{A}'\system{B}'$ and $U$ is a Haar-distributed random unitary on $\system{A}'\system{B}'$.
Then 
\begin{align*}
&\probability(\robustness(\randomshield\|\tau) \leq \upperconstant\sqrt{d}\log d)
\geq 1- e^{-\probconstant d^3 \log^2 d} 
\intertext{and consequently}
&\probability(\robustness(\randomshield)       \leq \upperconstant\sqrt{d}\log d )
\geq 1- e^{-\probconstant d^3 \log^2 d} \;.
\end{align*}
where $\upperconstant,\probconstant>0$ are universal constants (independent of $d$) and $\probability$ denotes the probability over the Haar-random choice.
\end{lemma}

\begin{proof}
    The second claim follows from the first by the definition in \Cref{eq:robustness}, so we only need to prove the upper bound estimate on $\robustness(\randomshield\|\tau)$.

    We use the same notation as in~\cite{ASY} and $\sepstates_0 = \sepstates - \tau$ as the set of separable states translated to the subspace of traceless Hermitian operators with the maximally mixed state at the origin. From the definition then we have that for any state $\varrho$
    \begin{align*}
    \robustness(\varrho\|\tau)
      &= \inf\quantity{s: \frac{1}{1+s}\quantity(\varrho+s\,\tau)\in\sepstates}
    \\&= \inf\quantity{s: \frac{1}{1+s}\quantity(\varrho-\tau) \in\sepstates_0}
    \, .
    \eqnumber
    \label{eq:robustness:tau}
    \end{align*}
    Let us with some abuse of notation denote by $\norm{\cdot}_{\sepstates_0}$ the gauge of $\sepstates_0$ (it is not homogeneous because $\sepstates_0$ is not symmetric, and hence is not actually a norm). 
    From \Cref{eq:robustness:tau}, we thus have that if $\varrho$ is entangled
    \begin{equation}
    \label{eq:robustness:sepgauge}
    \norm{\varrho-\tau}_{\sepstates_0}
    =\robustness(\varrho\|\tau) + 1
    \end{equation}
    (if $\varrho$ is separable this does not hold, as is the case for $\varrho=\tau$).

    Let $\bar\randomshield=P/\Tr P = 2P/d^2$, and let us introduce the notation $\varrho_0\equiv \varrho-\tau$ for any state $\varrho$. 
    Notice that $\randomshield_0 = U \bar\randomshield_0 U^\dagger$. 
    We reduced the problem of estimating $\robustness(\randomshield\|\tau)$
    to the problem of estimating $\norm{\randomshield_0}_{\sepstates_0}$ 
    and the statement to prove is thus
    \begin{align*}
    &\probability(
        \norm{U\bar\randomshield_0 U^\dagger}_{\sepstates_0} 
        \leq \upperconstant\sqrt{d}\log d)
    \geq 1- e^{-\probconstant d^3 \log^2 d}  \;.
    \end{align*}
    For this purpose, we first compute the expectation value $\E \norm{U\bar\randomshield_0 U^\dagger}_{\sepstates_0}$ over the random variable $U$. Then we estimate the Lipschitz constant of $\norm{U\bar\randomshield_0 U^\dagger}_{\sepstates_0}$ as a function of $U$ and use it to argue that being close to the the expected value happens with high probability. Notice that $\norm{X}_{\sepstates_0}$ is not unitary invariant, however the function $\E \norm{UXU^\dagger}_{\sepstates_0}$ is unitary invariant on $X$, while still being convex.
    
    Let us compute the expectation value  of $\E \norm{U\bar\randomshield_0 U^\dagger}_{\sepstates_0}$. With minor modifications, we know from~\cite[Lemma~6]{AL} that for any unitary-invariant convex function $g$ of any traceless Hermitian operators $X$ and $Y$ on $\C^d\otimes\C^d$, we have\footnote{The original~\cite[Lemma~6]{AL} is stated 
    only for permutation invariant norms. 
    However, what is proven in the proof is the more general statement 
    that if $x$ and $y$ are zero-sum vectors in $\R^d$, 
    then $x/\norm{x}_\infty \prec 2d \, y/\norm{y}_1$ 
    where $\prec$ denotes majorisation. 
    A direct application of~\cite[Theorem~II.3.3]{Bhatia} 
    then proves~\cite[Lemma~6]{AL} 
    and more generally $g(x)/\norm{x}_\infty \leq 2d \, g(y)/\norm{y}_1$ 
    for all permutation-invariant convex functions $g$. 
    Applying the latter to the spectrum of traceless Hermitian operators 
    for a unitary invariant function
    gives us \Cref{eq:traceless:convex}.}
    \begin{equation}
    \label{eq:traceless:convex}
    \frac{g(X)}{\norm{X}_\infty} \leq 2d^2 \frac{g(Y)}{\norm{Y}_1}
    \;,
    \end{equation}
    which applied twice leads to
    \begin{equation}
    \label{eq:traceless:convex:twice}
    \frac{1}{2d^2} \frac{\norm{X}_1}{\norm{Y}_\infty}
    \leq \frac{g(X)}{g(Y)}
    \leq 2d^2 \frac{\norm{X}_\infty}{\norm{Y}_1}
    \;.
    \end{equation}
    Now, we let $Y=\bar\randomshield_0$ for which $\norm{\bar\randomshield_0}_1=1$ and $\norm{\bar\randomshield}_\infty = 1/d^2$:
    \begin{equation*}
    \frac{1}{2} {\norm{X}_1}
    \leq \frac{g(X)}{g(\bar\randomshield_0)}
    \leq 2d^2 {\norm{X}_\infty}
    \;.
    \end{equation*}
    Then we let $g(X)=\E \norm{UXU^\dagger}_{\sepstates_0}$, which is unitary invariant by construction and convex by the convexity of $\norm{X}_{\sepstates_0}$:
    \begin{equation*}
    \frac{1}{2} {\norm{X}_1}
    \leq \frac{\E\norm{UXU^\dagger}_{\sepstates_0}}
              {\E\norm{U\bar\randomshield_0 U^\dagger}_{\sepstates_0}}
    \leq 2d^2 {\norm{X}_\infty}
    \;.
    \end{equation*}
    We now let $X$ be a Gaussian vector $G$ on the traceless Hermitian operators (Gaussian unitary ensemble) on $\C^d\otimes\C^d$.
    This makes $\E\norm{U\GUE U^\dagger}_{\sepstates_0}=\E\norm{\GUE}_{\sepstates_0}$.
    We then take the expectation values over the remaining random variable $\GUE$ on each side of the inequalities and get
    \begin{equation*}
    \frac{1}{4} \E{\norm{\GUE}_1}
    \leq \frac{\E\norm{\GUE}_{\sepstates_0}}{\E\norm{U\bar\randomshield_0 U^\dagger}_{\sepstates_0}}
    \leq 2d^2 \E{\norm{\GUE}_\infty}
    \;.
    \end{equation*}
    For such a Gaussian random matrix, it is well know that $\E\norm{\GUE}_1\sim d^3$ and $\E\norm{\GUE}_{\infty}\sim d$, where with ``$\sim$'' we denote having the same order. This proves
    \begin{equation}
    \E\norm{U\bar\randomshield_0 U^\dagger}_{\sepstates_0} \sim
    \E\norm{\GUE}_{\sepstates_0} /d^3
    \;.
    \end{equation}
    In particular, we know from~\cite[Section 4]{ASY} that $\E\norm{\GUE}_{\sepstates_0}$ is at most of order $d^{7/2}\log d$
    and therefore there exists a universal constant $\upperconstant>0$ such that 
    \begin{equation}
    \label{eq:gauge-S-expectation}
    \E \norm{\randomshield_0}_{\sepstates_0} 
    = \E\norm{U\bar\randomshield_0 U^\dagger}_{\sepstates_0}
    \leq \upperconstant\sqrt{d}\log d.
    \end{equation}
    
    Now we have to show that this average behaviour is generic for large $d$, because $f(U)\coloneqq\norm{U\bar\randomshield_0 U^\dagger}_{\sepstates_0}$ is regular enough. We claim that $f$
    is $4$-Lipschitz (in the Euclidean norm).
    Indeed, for any unitaries $U,V$ on $\C^d\otimes\C^d$, we have by the triangle inequality for $\norm{\cdot}_{\mathcal{S}_0}$
    \begin{align*}
    \quantity|f(U)-f(V)|
      &=
        \left|\norm{U\bar\randomshield_0 U^\dagger}_{\sepstates_0}
             -\norm{V\bar\randomshield_0 V^\dagger}_{\sepstates_0}\right| 
    \\&\leq
        \norm{U\bar\randomshield_0 U^\dagger 
            - V\bar\randomshield_0 V^\dagger}_{\sepstates_0}
    \\&=
        \norm{U\bar\randomshield U^\dagger 
            - V\bar\randomshield V^\dagger}_{\sepstates_0}
    \\&=
    \frac{2}{d^2} \norm{UPU^\dagger - VPV^\dagger}_{\sepstates_0}
    \; .
    \intertext{It was proven in~\cite{BG} that 
        \begin{equation*}
        \frac{1}{d\sqrt{d^2-1}} B_2 \subseteq \sepstates_0 \; ,
        \end{equation*}
        which implies
        \begin{equation*}
        \label{eq:sep0norm}
        \norm{\cdot}_{\sepstates_0} \leq d\sqrt{d^2-1}\norm{\cdot}_2 \leq d^2 \norm{\cdot}_2 \;.
        \end{equation*}
    Therefore we get}
    |f(U)-f(V)|
    &\leq 2 \norm{UPU^\dagger - VPV^\dagger}_2
    \\&\leq 2 \norm{(U-V)PU^\dagger}_2 + 2\norm{VP(U-V)^\dagger}_2
    \\&=    4 \norm{(U-V)P}_2
    \\&\leq 4 \norm{U-V}_2 \norm{P}_\infty
    \\&=    4 \norm{U-V}_2 \; .
    \end{align*}
    where the first inequality is by the triangle inequality and the last inequality is by H\"{o}lder inequality $\norm{XY}_2\leq\norm{X}_2\norm{Y}_{\infty}$.
    
    Now, we know from~\cite[Corollary 17]{MM} that any $L$-Lipschitz function $g$ on the unitaries on $\C^D$ (equipped with the Euclidean metric) satisfies the concentration estimate: if $U$ is a Haar-distributed unitary on $\C^D$, then for all $\epsilon>0$, $\probability( g(U) > \E g + \epsilon ) \leq e^{-{\probconstant}D \epsilon^2/L^2}$, where $\probconstant>0$ is a universal constant. Combining the above estimate on the Lipschitz constant of $f$ with the estimate on its expected value from \Cref{eq:gauge-S-expectation}, we thus get for all $\epsilon>0$
    \begin{align*}
    \probability({f}(U) > \upperconstant\sqrt{d}\log d+\epsilon) 
    &\leq
    \probability\big({f}(U) > \E{f}+\epsilon) 
    \\&\leq 
    e^{-\probconstant' d^2\epsilon^2}.
    \end{align*}
    The advertised result follows from choosing $\epsilon=\upperconstant\sqrt{d}\log d$ in the above deviation probability, and combining it with \Cref{eq:robustness:sepgauge} .
\end{proof}

\begin{corollary}
\label{lemma:robustness:randomplusminus}
Let $\randomshieldplusminus$ be random orthogonal states on $\system{A}'\system{B}'$ as in \Cref{eq:def-rho-sigma'}.
Then
\begin{align*}
&\probability(\max\quantity
    {\robustness(\randomshieldplus),\robustness(\randomshieldminus)}
    \leq {\upperconstant}{\sqrt{d}\log d})
\geq 1- e^{-\probconstant d^3 \log^2 d} .
\end{align*}
where $\upperconstant,\probconstant>0$ are universal constants.
\end{corollary}

\begin{proof}
This is a direct consequence of \Cref{lemma:robustness:random}. By the union bound, we have 
\begin{align*}
& \probability(\robustness(\randomshieldplus)       \geq {\upperconstant}{\sqrt{d}\log d}\ \text{or}\ \robustness(\randomshieldminus)       \geq {\upperconstant}{\sqrt{d}\log d}) \\
& \leq \probability(\robustness(\randomshieldplus)       \geq {\upperconstant}{\sqrt{d}\log d} ) +  \probability(\robustness(\randomshieldminus)       \geq {\upperconstant}{\sqrt{d}\log d}) \\
& \leq 2e^{-\probconstant d^3 \log^2 d}
\qedhere
\end{align*}
\end{proof}

It was proved in~\cite[Section 6.1]{AL} that such random orthogonal states $\randomshieldplusminus$ have interesting data-hiding properties, namely: with high probability SEP (and even more so LOCC) measurements almost do not distinguish them, while PPT ones do. So our goal is now to construct random private states out of them, and show that they exhibit some similar features. 

\begin{construction}[Random quantum private state]
\label{eq:def-random-gamma} 
Let $|\system{A}|=|\system{B}|=2$ and let $\randomshieldplusminus$ be two random orthogonal states $\system{A}'\system{B}'$ as in \Cref{eq:def-rho-sigma'}. We define a random private state $\gamma$ on $\system{A}\system{A}'\system{B}\system{B}'$ to be a private state as in \Cref{eq:def-gamma} with $\shieldplusminus=\randomshieldplusminus$, namely
\begin{equation*} 
\gamma_{\system{A}\system{A}'\system{B}\system{B}'} 
= \frac{1}{2}\left(\psi^+_{\system{AB}}\otimes\randomshieldplus_{\system{A}'\system{B}'} 
+\psi^-_{\system{AB}}\otimes\randomshieldminus_{\system{A}'\system{B}'}\right) \, .
\end{equation*}
\end{construction}

Observe that, for such construction the key-attacked state is always separable. Indeed, since $P+P^{\perp}=\openone$, we simply have $(\randomshieldplus+\randomshieldminus)/2=\openone/d^2$ and therefore $\hat{\gamma}=\quantity(\psi^+ + \psi^-)/2\otimes\openone/d^2$. 
This is thus an irreducible private state, namely the distillable key is exactly one bit~\cite{HHHOb}.

\begin{note*}
Irreducible private states are those private states $\gamma$ on ${\system{A}\system{A}'\system{B}\system{B}'}$
for which $K_D(\gamma)=\log|\system{A}|$~\cite{HHHOb}.
Private states for which $\hat\gamma\in\sepstates(\system{A}\system{A}'{:}\system{B}\system{B}')$
are always irreducible, and they are sometimes called strictly irreducible~\cite{CF,horodecki2016irreducible}.
Whether all irreducible private states are also strictly irreducible depends on whether there exist so called bound key states~\cite{horodecki2016irreducible}, entangled states that have zero distillable key. The existence of bound key is a long-standing open question even in a classical setting (see \cite{ozols2014bound} and references there in) .
\end{note*}

Finally, note that $\randomshieldplusminus$ are orthogonal by construction but, as we already know from~\cite[Section 6.1]{AL}, this orthogonality is completely hidden to local observers.
In the remainder of this work we will extend this result to the case of random private states.

\subsection{Distinguishability of random private quantum states from separable states} 

We start with a statement on the distinguishability of our random private states from separable states, measured in local norms.

\begin{theorem}
\label{th:gamma-SEP}
\label{th:gamma-PPT}
Let $\gamma$ be a random private state 
on $\system{A}\system{A}'\system{B}\system{B}'$ as defined by \Cref{eq:def-random-gamma}. Then,
\begin{align*}
& \probability(\norm{\gamma-\sepstates(\system{A}\system{A}'{:}\system{B}\system{B}')}_{\sep(\classic{A}\classic{A}'{:}\classic{B}\classic{B}')} 
\leq \frac{\upperconstant}{\sqrt{d}}) \geq 1-e^{-{\probconstant}d^3}
 \\
& \probability(\norm{\gamma-\sepstates(\system{A}\system{A}'{:}\system{B}\system{B}')}_{\ppt(\classic{A}\classic{A}'{:}\classic{B}\classic{B}')} 
\geq \lowerconstant ) \geq 1-e^{-{\probconstant}d^4}
\end{align*}
where $\probconstant,\lowerconstant,\upperconstant>0$ are universal constants.
\end{theorem}

\begin{proof}
For the first claim, we know from~\cite[Section~6.1]{AL} that there exist universal constants $\probconstant,\upperconstant>0$ such that,
\[ \probability\big(\norm{\randomshieldplus-\randomshieldminus}_{\sep} 
 \leq {\upperconstant}/{\sqrt{d}} )  \geq 1-e^{-{\probconstant}d^3}.\]
Hence by \Cref{lemma:pbit-M-norms-upper} above, we have (just relabelling $3 \upperconstant/2$ as $\upperconstant$), that
\[ \probability\big(\norm{\gamma-\sepstates}_{\sep} \leq {\upperconstant}/{\sqrt{d}} ) \geq 1-e^{-{\probconstant}d^3}.\] 
For the second claim, we know from~\cite[Theorem 5]{AL}  
that there exist universal constants $\probconstant,\lowerconstant>0$ such that 
\[ \probability(\norm{\randomshieldplus-\randomshieldminus}_{\ppt} \geq \lowerconstant) 
\geq 1-e^{-{\probconstant}d^4}.  \]
Hence by \Cref{lemma:pbit-M-norms-lower} above, we have (just relabelling $\lowerconstant/3$ as $\lowerconstant$),
\[ \probability(\norm{\gamma-\sepstates}_{\ppt} \geq \lowerconstant) 
\geq 1-e^{-{\probconstant}d^4}. \]
This concludes the proof of \Cref{th:gamma-PPT}.
\end{proof}

In words, \Cref{th:gamma-SEP} tells us the following: the considered random private state $\gamma$ is, with probability going to $1$ as the dimension $d$ grows, at a SEP-norm distance of at most $\upperconstant/\sqrt{d}$ and at a PPT-norm distance of at least $\lowerconstant$ from the set of separable states. So in conclusion, what we learn from it is that there exist private states which are barely distinguishable from being separable for observers which can only perform SEP (and even more so LOCC) measurements on them. However, this data hiding property is not maintained when relaxing to PPT operations, since these private states keep a constant distinguishability from separable states under PPT measurements. 
\smallskip

We now derive an analogue of \Cref{th:gamma-SEP} when distinguishability is measured in local relative entropy.

\begin{theorem}
\label{th:D_SEP-D_PPT}
Let $\gamma$ be a random private state
on ${\system{A}\system{A}'\system{B}\system{B}'}$ as defined by \Cref{eq:def-random-gamma}. 
Then
\begin{align*} 
&\probability(\!
    D_{\sep(\classic{A}\classic{A}'{:}\classic{B}\classic{B}')}
    \quantity(\gamma \|\sepstates(\system{A}\system{A}'{:}\system{B}\system{B}'))
    \leq\upperconstant\frac{\log d}{\sqrt{d}}) \!
 \geq 1 - e^{-{\probconstant}d}
\\
& \probability(
    D_{\ppt(\classic{A}\classic{A}'{:}\classic{B}\classic{B}')}
    \quantity(\gamma\|\sepstates(\system{A}\system{A}'{:}\system{B}\system{B}'))
    \geq \lowerconstant)  
 \geq 1-e^{-{\probconstant}d^4} 
\end{align*}
where $\probconstant,\lowerconstant,\upperconstant>0$ are universal constants.
\end{theorem}

The upper bound on $D_{\sep}(\gamma\|\sepstates)$ in \Cref{th:D_SEP-D_PPT} above is not tight: the $\log d$ factor can actually be removed. The derivation of this improved upper bound is relegated to Appendix \ref{appendix:improvement} as it is much more involved, and requires developing several additional tools (which might be of independent interest). 

\begin{proof}
For the first probability estimate, we know from \Cref{th:gamma-SEP} that, with probability greater than $1-e^{-{\probconstant}d}$, $\norm{\gamma-\hat{\gamma}}_{\sep}\leq \upperconstant/\sqrt{d}$. 
Hence by the asymptotic continuity of the measured relative entropy \cite{LiW} we get that, with probability greater than $1-e^{-{\probconstant}d}$,
\begin{align*} 
\label{eq:D_SEP-proof}
D_{\sep}\left(\gamma\|\sepstates\right) &= \quantity| D_{\sep}(\gamma\|\sepstates) -  D_{\sep}(\hat{\gamma}\|\sepstates)| 
\\&\leq \frac{\upperconstant}{\sqrt{d}}\log (2d) + g\quantity(\frac{\upperconstant}{\sqrt{d}})
\\&\leq\frac{\upperconstant'\log d}{\sqrt{d}} 
\, ,
\end{align*}
where the equality is due to $\hat{\gamma}\in\sepstates$, so that $D_{\sep}(\hat{\gamma}\|\sepstates)=0$. 

For the second probability estimate, we know from \Cref{th:gamma-PPT} that, with probability greater than $1-e^{-{\probconstant}d^4}$, $\|\gamma-\sepstates\|_{\ppt} \geq \lowerconstant$. By Pinsker's inequality, this implies that, with probability greater than $1-e^{-{\probconstant}d^4}$,
\[ D_{\ppt}(\gamma\|\sepstates)\geq \frac{1}{2\ln 2}\|\gamma-\sepstates\|_{\ppt}^2 \geq \frac{\lowerconstant^2}{2\ln 2}=\lowerconstant'\, .\]
Notice, however, that while Pinsker's inequality suffices for qubit key systems, it scales badly for higher dimensional key systems. A scalable version of this proof can be obtained proving \Cref{lemma:pbit-M-norms-lower} directly for the relative entropy using \Cref{lemma:iso:distinguishability}
\end{proof}

\Cref{th:D_SEP-D_PPT} teaches us that the same qualitative conclusion as that of \Cref{th:gamma-SEP} holds when measuring distance from the set of separable states in relative entropy rather than trace norm: with probability going to $1$ as the dimension $d$ grows, our random private state $\gamma$ has  a very small relative entropy of entanglement when restricted to SEP (and even more so LOCC) measurements, but a high one when only restricted to PPT measurements.

\section{Quantum key repeater}
\label{sec:repeater}

\newcommand{\alice}{{\smash{\system{A}'}}}
\newcommand{\bob}  {{\smash{\system{B}'}}}
\newcommand{\ab} {{\alice{:}\bob}}
\newcommand{\atb}{{A{\to}B}}

\newcommand{\charlieC}{\system{C}}
\newcommand{\charlieA}{{{\charlieC}}}
\newcommand{\charlieB}{{{\tilde{\charlieC}}}}
\newcommand{\rhoA}{\varrho} 
\newcommand{\rhoB}{\tilde{\rhoA}}
\newcommand{\rhoAsystems}{\varrho$ on ${\vphantom{\charlieB}\alice\charlieA}} 
\newcommand{\rhoBsystems}{\rhoB$ on ${\charlieB\bob}}
\newcommand{\charlie} {{\charlieA\charlieB}}
\newcommand{\ac}  {{\alice{:}\charlieA}}
\newcommand{\cb}  {{\smash\charlieB{:}\bob}}
\newcommand{\accb}{{\alice{:}\charlieA{:}\charlieB{:}\bob}}
\newcommand{\acb} {{\alice{:}\charlie{:}\bob}}
\newcommand{\abc} {{\alice\bob{:}\charlie}}
\newcommand{\cab} {{\charlie{:}\alice\bob}}
\newcommand{\abtc}{{\alice\bob\leftrightarrow\charlie}}
\newcommand{\ctab}{{\charlie\rightarrow\alice\bob}}
\newcommand{\abfc}{{\alice\bob\leftarrow\charlie}}
\newcommand{\sabc}{{\alice\bob\mathrlap{{:}\charlie}}}
\newcommand{\repcutA}{{\alice{:}\charlie\bob}}
\newcommand{\repcutB}{{\alice\charlie{:}\bob}}
\newcommand{\trc}{\Tr_{\scriptscriptstyle\charlie}}

\newcommand{\tensor}[1]{^{\otimes{#1}}}
\newcommand{\KD}{K_D\quantity}

\newcommand{\textcopies}  {one}
\newcommand{\copies}      {}
\newcommand{\Rcopies}     {1}
\newcommand{\fraccopies}  {}
\newcommand{\scriptcopies}{}
\newcommand{\tensorcopies}{}

\subsection{$\Rcopies$-bounded quantum key repeater rate} 

Ultimately, we would like to have a full understanding 
of key distillation in a general network scenario.
The immediate first step toward this is to add 
just a single intermediate station Charlie between Alice and Bob. 
This setting, known as the quantum key repeater, has been introduced in~\cite{BCHW}. 
The inputs are now two states $\rhoA$ and $\rhoB$, the first one shared between Alice and Charlie and the second one shared between Charlie and Bob, assumed to be produced independently and thus to be in tensor product.
As in entanglement distillation and key distillation, the parties are given arbitrarily many of copies, for instance produced by identical independent use of a quantum channel.
Using only LOCC operations with public communication, classical communication such that the eavesdropper obtains a classical copy of all the classical data exchanged between the parties, the parties are supposed to distill as close to and as many bits of key as they can, where the task is made possible because Charlie can act globally on his parts of the states.
However, while Charlie is essential to achieve the goal, he is also untrusted,
therefore the key must be secret also from Charlie. Equivalently, we say that Charlie's systems are given to the eavesdropper at the end of the protocol.
In~\cite{HHHO} it was proven that distilling key is equivalent to distilling private states with plain LOCC, without keeping track of the state of the eavesdropper and the public communication copied to it.
In this framework, the eavesdropper becomes the environment, and giving systems to the eavesdropper becomes the trace operation.
The class of LOCC protocols that end tracing Charlie is denoted by $\locc^\leftrightarrow_C(\quantum{A}{:}\quantum{B})=\Tr_{\system{C}}\circ\locc(\quantum{A}{:}\quantum{C}{:}\quantum{B})$. Furthermore, the class of protocols, where Charlie first measures and sends his result to Alice and Bob, who then followup with an arbitrary two-way LOCC protocol is denoted by 
$\locc^\rightarrow_C(\quantum{A}{:}\quantum{B}).$

Like~\cite{BCHW} and~\cite{CF}, 
we consider a variation of the quantum key repeater rate, namely
the $\Rcopies$-bounded quantum key repeater rate, for which associated rates $R_D^{\Rcopies, \leftrightarrow}$ and $R_D^{\Rcopies, \rightarrow}$ can be defined.
The operational interpretation goes as follows: instead of letting Charlie act jointly on arbitrarily many copies of the input states,
we restrict him to act only on \textcopies.
This should model, for example, bounded-memory repeater stations
that can only act on a finite number of copies at the same time.
Alice and Bob can still apply their distillation protocols without restriction on the outcomes of the operations with Charlie. In Appendix \ref{appendix:R_D}, we improve some upper bounds on the general quantum key repeater setting described above, where Charlie can act jointly on arbitrary many copies of the input states.

\begin{definition}[$\Rcopies$-bounded quantum key repeater rate] For any states $\rho$ on $AC$ and $\tilde\rho$ on $\tilde{C}B$, we define
    \begin{align*}
    R_D^{\Rcopies, \leftrightarrow} (\rhoA,\rhoB) &\coloneqq \fraccopies
    \lim _{\varepsilon\to0} \lim _{n\to\infty}
    \sup _{\Pi,\Lambda}
    \quantity{r: \Pi\quantity\big(\quantity[\Lambda\quantity(\rhoA\tensorcopies\otimes\rhoB\tensorcopies)]^{\otimes n}) \approx_\varepsilon \gamma^{ rn }}
    \end{align*}
where $\Pi\in\locc(\quantum{A'}^{n\copies}{:}\quantum{B'}^{n\copies})$ and 
	$\Lambda\in \locc^\leftrightarrow_{C\tilde{C}}(\quantum{A}{:}\quantum{B})$, where $A'$ and $B'$ are the output systems of $\Lambda$. $\approx_\epsilon$ is understood in the trace distance. We make a similar definition with $\leftrightarrow$ replaced by $\rightarrow$. $\gamma^{rn}$ denotes a private state with $\lfloor rn \rfloor$ qubits in each parties key system.
\end{definition}

In~\cite{BCHW} and~\cite{CF}, the $\Rcopies$-bounded quantum key repeater rate was defined using the equivalent expression in \Cref{lem:R_D-K_D} below, which will be more convenient to use in the next section.
\Cref{lem:R_D-K_D} rewrites $R_D^\Rcopies (\rhoA,\rhoB)$ as an optimization over bipartite distillable keys. 
Therefore let us recall that the distillable key $K_D$ of a state $\varrho$ on ${\system{A}\system{B}}$ is defined as~\cite{HHHO}
\[ K_D(\varrho) \coloneqq \fraccopies
\lim _{\varepsilon\to0} \lim _{n\to\infty} 
\quantity{K: \Pi\quantity(\varrho^{\otimes n}) \approx_\varepsilon \gamma^{Kn}}
\]
where $\Pi\in\locc(\quantum{A}^{n\copies}{:}\quantum{B}^{n\copies})$.

\begin{lemma} \label{lem:R_D-K_D}
    For any states $\rho$ on $AC$ and $\tilde\rho$ on $\tilde{C}B$, we have
    \[ R_D^{\Rcopies, \leftrightarrow} (\rhoA,\rhoB) = \fraccopies \sup _{\Lambda} 
    \KD ( \Lambda\quantity(\rhoA\tensorcopies\otimes\rhoB\tensorcopies) )\, .
    \] 
    where $\Lambda\in\locc^{\leftrightarrow}_{C\tilde{C}}
    (\quantum{A}{:}\quantum{B})$. A similar statement holds with $\leftrightarrow$ replaced by $\rightarrow$. 
\end{lemma}

\begin{proof}
We will only prove the statment for $\leftrightarrow$, the other statement will follow mutatis mutandis.
    The inequality $$\KD(\Lambda\quantity({\rhoA}\tensorcopies\otimes{\rhoB}\tensorcopies))
    \leq \copies R_D^{\Rcopies, \leftrightarrow} (\rhoA,\rhoB)$$ is trivial for $\Lambda\in\locc^{\leftrightarrow}_{C\tilde{C}}
    (\quantum{A}{:}\quantum{B})$. So let us now concentrate on showing the  other direction.
    By definition, any $R <\copies R_D^{\Rcopies, \leftrightarrow}(\rhoA,\rhoB)$ is an achievable rate,
    and thus 
    for all $\varepsilon>0$, there exists $n,$ $\Lambda$ and $\Pi$ as in the definition of the rate satisfying
    \begin{align*}
    \Pi\quantity([\Lambda\quantity(\rhoA\tensorcopies\otimes\rhoB\tensorcopies)]\tensor n ) \approx_\varepsilon \gamma^{nR} \, .
    \end{align*}
    However, it was shown in \cite[Lemma 6]{HHHOb}
    that if $\sigma\approx_\varepsilon \gamma^k$,
    then $\KD(\sigma)\geq k - 4\varepsilon k - (2-\varepsilon)h(\varepsilon)$,
    where we used the improved bounds from \cite{Winter,Shirokov}
    (which can be done by using the continuity of the entropy to
    estimate the Devetak-Winter rate of $\sigma$).
    This means that we have
    \begin{align*}
    nR - 4\varepsilon nR - 2 h(\varepsilon) 
    &\leq \KD(\Pi\quantity(\Lambda\quantity(\rhoA\tensorcopies\otimes\rhoB\tensorcopies)\tensor n ))
    \\&\leq \KD(\Lambda\quantity(\rhoA\tensorcopies\otimes\rhoB\tensorcopies)\tensor n )
    \\&= n \KD(\Lambda\quantity(\rhoA\tensorcopies\otimes\rhoB\tensorcopies) )
    \, ,
    \end{align*}
    where we used the monotonicity of the key rate in the second inequality and the asymptotic definition of the key rate in the equality.       
    Thus for any $\varepsilon>0$ there exists $\Lambda$ such that 
    \[\KD(\Lambda\quantity({\rhoA}\tensorcopies\otimes{\rhoB}\tensorcopies)) \geq (1-4\varepsilon)\copies R - 2 h(\varepsilon).\]
Taking the infimum over $\epsilon$ and the supremum over $R$ gives the desired result.
\end{proof}

\subsection{Upper bound on the $\Rcopies$-bounded quantum key repeater rate of random private quantum states}

We can now provide an upper bound on $R_D^{\Rcopies, \leftrightarrow} $ and $R_D^{\Rcopies, \leftarrow} $ in terms of an LOCC-restricted relative entropy distance. 

\begin{proposition}
    \label{lemma:repeatersep-singlecopy}
    Let $\states\in\allstates(\system{A\charlieA\charlieB B})$
    be a set of states such that
    $\Lambda\in\locc_{C\tilde{C}}^\leftrightarrow(\quantum{A}{:}\quantum{B})$ implies $\Lambda(\states)\subseteq\sepstates(\ab)$, where $A'$ and $B'$ are the outputs of the map $\Lambda$.
    Then, for any states $\rho$ on $AC$ and $\tilde\rho$ on $\tilde{C}B$, we have
    \begin{align*}
    R_D^{\Rcopies, \leftrightarrow} (\rhoA,\rhoB) 
    \leq  D_{\locc^{\leftrightarrow}_{C\tilde{C}}(\quantum{A}{:}\classic{\charlieA}\classic{\charlieB}{:}\quantum{B})}\scriptcopies \quantity\big
    ( \rhoA\otimes\rhoB\| \states ).
    \end{align*}
A similar statement holds with $\leftrightarrow$ replaced by $\rightarrow$. \end{proposition}

\begin{proof}
    \allowdisplaybreaks
    By assumption, for any $\Lambda\in\locc_{C\tilde{C}}^\leftrightarrow(\quantum{A}{:}\quantum{B})$ and $\sigma\in\states$, we have $\Lambda(\sigma)\in\sepstates(\alice\scriptcopies{:}\bob\scriptcopies)$. 
    Therefore, for any such $\sigma$,
    \begin{align*}
    D&_{\locc_{C\tilde{C}}^\leftrightarrow(\quantum{A}{:}\quantum{B})}(\rhoA\tensorcopies\otimes\rhoB\tensorcopies\parallel \sigma ) \\
    & = \sup_{\Lambda\in\locc_{C\tilde{C}}^\leftrightarrow(\quantum{A}{:}\quantum{B})}
    D\quantity\big(\Lambda(\rhoA\tensorcopies\otimes\rhoB\tensorcopies)\parallel \Lambda(\sigma ))  \\
    &\geq \sup_{\Lambda\in\locc_{C\tilde{C}}^\leftrightarrow(\quantum{A}{:}\quantum{B})}
    D\quantity\big(\Lambda(\rhoA\tensorcopies\otimes\rhoB\tensorcopies)\parallel \sepstates (\alice\scriptcopies{:}\bob\scriptcopies)) \\
    & \geq \sup_{\Lambda\in\locc_{C\tilde{C}}^\leftrightarrow(\quantum{A}{:}\quantum{B})}
    \KD\big(\Lambda(\rhoA\tensorcopies\otimes\rhoB\tensorcopies)) \\
    & = \copies R_D^\Rcopies (\rhoA,\rhoB)),
    \end{align*}
    where the last equality is due to \Cref{lem:R_D-K_D} and the last inequality is due to $D(\cdot\|\sepstates)\geq K_D$, a known upper bound on the distillable key~\cite{HHHO}. Taking the infimum over $\sigma$ ends the proof.
\end{proof}

In order to facilitate the use of the proposition, we will record the following corollary.

\begin{corollary}
\label{cor:repeatersep-singlecopy}
For any states $\rho$ on $AC$ and $\tilde\rho$ on $\tilde{C}B$, we have
    \begin{align*}
    R_D^{\Rcopies, \leftrightarrow} (\rhoA,\rhoB) 
    \leq  D_{\locc(\quantum{A}\quantum{B}{:}\classic{\charlieA}\classic{\charlieB})}\scriptcopies \quantity\big
    ( \rhoA\otimes\rhoB\| {\cal S}(A{:}C\tilde{C}B)).
    \end{align*}
and 
    \begin{align*}
    R_D^{\Rcopies, \rightarrow} (\rhoA,\rhoB) 
    \leq  D_{\allmeasurements(C\tilde{C})}\scriptcopies \quantity\big
    ( \rhoA\otimes\rhoB\| {\cal S}(A{:}C\tilde{C}B)).
    \end{align*}
\end{corollary}
\begin{proof}
Note that ${\cal S}(A{:}C\tilde{C}B)$ satisfies the conditions of the proposition for a set ${\cal K}$. Note further that the bound in the proposition remains unchanged if we omit to trace out the classical system at Charlie's side. Hence, we find
    \begin{align*}
    R_D^{\Rcopies, \leftrightarrow} (\rhoA,\rhoB) 
    \leq  D_{\locc(\quantum{A}{:}{C\tilde{C}}{:}\quantum{B})}\scriptcopies \quantity\big
    ( \rhoA\otimes\rhoB\| {\cal S}(A{:}C\tilde{C}B)).
    \end{align*}
Since $\locc(\quantum{A}{:}{C\tilde{C}}{:}\quantum{B}) \subseteq \locc(\quantum{A}\quantum{B}{:}\classic{\charlieA}\classic{\charlieB})$ the first claim follows.

In order to obtain the second bound, we again note that ${\cal S}(A{:}C\tilde{C}B)$ satisfies the conditions of the proposition. Note that any $\Lambda\in\locc_{C\tilde{C}}^\rightarrow(\quantum{A}{:}\quantum{B})$ is equivalently given by a local measurement on charlie, followed by some other global quantum channel and measurement. We remove that quantum channel by use of the monotonicity of the quantum relative entropy and obtain the final statement. 
\end{proof}

Thanks to these upper bounds we are now in position to establish --  partly based on a conjecture -- that two random private states have with high probability a small $\Rcopies$-bounded quantum key repeater rate. 

\renewcommand{\alice}{\system{A}}
\renewcommand{\bob}{\system{B}}

\begin{theorem}
\label{th:random-R1}
Let $|\alice'|=|\charlieA'|=d$ and $|\charlieB|=|\bob|=k$,
and let $\gamma$ be a random private state on ${\alice\alice'\charlieA\charlieA'}$
as defined by \Cref{eq:def-random-gamma}. 
Then,
for any state $\varrho$ on $\charlieB\bob$ it holds
\begin{equation*} 
\label{eq:random-R1} 
\probability(R_D^{1, \leftrightarrow}(\gamma,\varrho) 
\leq \upperconstant \frac{k}{\sqrt{d}}\log{d} )
\geq 1 - e^{-\probconstant d^3}\, ,
\end{equation*}
if Conjecture \ref{prop:D_L|K} is true ($\probconstant,\upperconstant>0$ are constants). Unconditionally, we have
\begin{equation*} 
\label{eq:random-R1} 
\probability(R_D^{1, \rightarrow}(\gamma,\varrho) 
\leq \upperconstant \frac{k}{\sqrt{d}}\log{d} )
\geq 1 - e^{-\probconstant d^3}\, .
\end{equation*}
\end{theorem}

\begin{proof}
We estimate
\begin{align} 
&  \norm{(\gamma-\hat{\gamma})\otimes\varrho}
_{\locc(\quantum{A}\quantum{A}'\quantum{B}{:}\classic{C}\classic{C}'\tilde{\classic{C}})} 
\nonumber
\\& 
\label{eq:loccACB-sepABC}
\leq \norm{(\gamma-\hat{\gamma})\otimes\varrho}
_{\sep(\classic{A}\classic{A}'\classic{B}{:}\classic{C}\classic{C}'\tilde{\classic{C}})} .
\end{align}

Now let $\Delta=\randomshieldplus-\randomshieldminus$, where $\randomshieldplusminus$ are the random orthogonal states on $\system{A}'\system{C}'$ appearing in the shield of $\gamma$. 
Observe that  by \Cref{prop:psi-M}
\begin{align*}
\norm{(\gamma-\hat{\gamma})\otimes\varrho}
&_{\sep(\classic{AA'B{:}CC'\tilde{C}})} 
\\& = \frac{1}{4} 
    \norm{(\psi^+-\psi^-)\otimes\Delta\otimes\varrho}
    _{\sep(\classic{AA'B{:}CC'\tilde{C}})}
\\& \leq \frac{1}{2}
    \norm{\psi^+\otimes \Delta\otimes\varrho }
    _{\sep(\classic{AA'B{:}CC'\tilde{C}})}
\\& \leq \frac{1}{2} 
    (2\robustness(\psi^+\otimes\varrho)+1)
    \norm{\Delta}
    _{\sep(\classic{A}'{:}\classic{C}')}
\\& \leq \frac{1}{2} (4k-1)  
    \norm{\Delta}
    _{\sep(\classic{A}'{:}\classic{C}')}
\\& \leq 2 k  
    \norm{\Delta}_{\sep(\classic{A}'{:}\classic{C}')}.
\end{align*}
Yet, we know from \cite[Section 6.1]{AL}, that $\norm{\Delta}_{\sep}\leq \upperconstant'/\sqrt{d}$ with probability greater than $1-e^{-\probconstant d^3}$. 
Therefore, joining with \Cref{eq:loccACB-sepABC}, we get
\begin{align} 
\label{eq:random-normpartial} 
&\probability( 
\norm{(\gamma-\hat{\gamma})\otimes\varrho}
_{\locc(\quantum{A}\quantum{A}'\quantum{B}{:}\classic{C}\classic{C}'\tilde{\classic{C}})} 
\leq \upperconstant\frac{k}{\sqrt{d}} ) \\
&\qquad \qquad  \geq 1 - e^{-\probconstant d^3}. \nonumber
\end{align}

{We can then apply Conjecture \ref{prop:D_L|K} 
and use the fact that $\hat{\gamma}\otimes\varrho\in\sepstates
(\alice\alice'{:}\charlieA\charlieA'\charlieB\bob)$  to find
\begin{align*}
D&_{\locc(\quantum{A}\quantum{B}{:}\classic{\charlieA}\classic{\charlieB})} 
\quantity(\gamma\otimes\varrho \middle\| 
\sepstates(\alice\alice'{:}\charlieA\charlieA'\charlieB\bob))
\\&=
\left|
D_{\locc(\quantum{A}\quantum{B}{:}\classic{\charlieA}\classic{\charlieB})} 
\quantity(\gamma\otimes\varrho \middle\| 
\sepstates(\alice\alice'{:}\charlieA\charlieA'\charlieB\bob))
\right.
\\&\quad-
\left.
D_{\locc(\quantum{A}\quantum{B}{:}\classic{\charlieA}\classic{\charlieB})} 
\quantity(\hat\gamma\otimes\varrho \middle\| 
\sepstates(\alice\alice'{:}\charlieA\charlieA'\charlieB\bob))
\right|
\\&\leq \kappa 
\norm{(\gamma-\hat{\gamma})\otimes\varrho}
_{\locc(\quantum{A}\quantum{B}{:}\classic{\charlieA}\classic{\charlieB})} \log d
\\&\quad+ g\quantity(\norm{(\gamma-\hat{\gamma})\otimes\varrho}
_{\locc(\quantum{A}\quantum{B}{:}\classic{\charlieA}\classic{\charlieB})})
\,.
\end{align*}
By Corollary \ref{cor:repeatersep-singlecopy} 
    \begin{align*}
    R_D^{\Rcopies, \leftrightarrow} (\gamma, \rho) 
    \leq  D_{\locc(\quantum{A}\quantum{B}{:}\classic{\charlieA}\classic{\charlieB})}\scriptcopies \quantity\big
    ( \gamma\otimes\rho\| {\cal S}(A{:}C\tilde{C}B)).
    \end{align*}

Combining \Cref{eq:random-normpartial} with these two last inequalities leads us to our first claim (with a suitable new constant $C$)
Similarly, using Proposition \ref{prop:D_L|K0} instead of \Cref{prop:D_L|K}} results in a proof of the second claim.
\end{proof}

\subsection[Relation to the PPT square conjecture]{Relation to the PPT\textsuperscript{2} conjecture} 

Motivated by the question whether the key repeater rate can be zero for states with non-zero key rate, the \emph{PPT\textsuperscript{2} conjecture} was introduced by the first author in~\cite{PPT^2}. It states that the sequential composition of any PPT channel is entanglement-breaking (see \Cref{sec:bipartite} for the difference between PPT channels and PPT operations). If true, it would imply that the key repeater rate of two PPT states (which are the only known examples for bound entangled states) is zero, including if the two states are bound entangled states with key. That is, if $\varrho, \tilde \varrho$ are PPT states, then $R_D(\varrho, \tilde \varrho)=0$ even if $K_D(\varrho),K_D(\tilde \varrho)>0$.

Whereas some progress has been made on the conjecture, the general case is still open.
It has subsequently been shown that the gap between key and repeated key can be made arbitrarily large for certain bound entangled states. In~\cite{BCHW} the examples were based on noisy private bit constructions. While in~\cite{CF} large gaps have been shown for noiseless private bits for key repeaters with one-way communication from the repeater station. Even though noiseless private states are NPT (non-positive under partial transposition), the upper bounds in concrete examples have mostly been based on the partial transposition, e.g.~bounding the log-negativity of the states (see \Cref{sec:remarks} for further comments on this point). 

In this work, we presented examples of private bits that have a large gap between key rate and $\Rcopies$-bounded key repeater rate. Our examples are carefully constructed so that arguments based on the partial transposition do not immediately apply (their PPT-restricted relative entropy distance from separable states is large). Our results thus give a complementary view on the PPT\textsuperscript{2} conjecture, by providing non-PPT channels, whose sequential composition has low key, a property shared by entanglement-breaking channels. Our work might thus be viewed as pointing to extensions of the PPT\textsuperscript{2} conjecture.

We would also like to mention an implication for the older \emph{NPT bound entanglement conjecture}, which postulates that there exist undistillable NPT states~\cite{BCDL,DSSTT}. Since the states that we have constructed have low LOCC-restricted relative entropy of entanglement, and since the regularised LOCC-restricted relative entropy of entanglement is an upper bound on the distillable entanglement, they are likely to have small distillable entanglement. We note that the constructed states have large log-negativity and are thus not close to being PPT.

\section{Miscellaneous remarks}
\label{sec:remarks}

We have shown that 
$D_{\sep(\classic{A}\classic{A}'{:}\classic{B}\classic{B}')}
\left(\gamma\|\sepstates(\system{A}\system{A}'{:}\system{B}\system{B}')\right)$
is small with high probability for our random private state $\gamma$ on
${\system{A}\system{A}'\system{B}\system{B}'}$. 
However, what we would ultimately like to show is 
that this remains true for the associated regularized quantity 
$D_{\sep(\classic{A}\classic{A}'{:}\classic{B}\classic{B}')}^{\infty}
\left(\gamma_{\system{A}\system{A}'\system{B}\system{B}'}\|\sepstates(\system{A}\system{A}'{:}\system{B}\system{B}')\right)$, 
whose definition we recall now. Given $\varrho$ a state on 
$\system{C}\system{D}$, first define for each $n\in\N$
\[ D_{\sep(\classic{C}{:}\classic{D})}^{n}\!
\left(\varrho\|\sepstates(\system{C}{:}\system{D})\right) 
\coloneqq  \frac{1}{n} 
D_{\sep(\classic{C}^n{:}\classic{D}^n)}\!
\left(\varrho^{\otimes n}\|\sepstates(\system{C}^n{:}\system{D}^n)\right) 
\]
and then taking the limit, define
\[ D_{\sep(\classic{C}{:}\classic{D})}^{\infty}
\left(\varrho\|\sepstates(\system{C}{:}\system{D})\right) 
\coloneqq  \lim_{n\rightarrow+\infty}  D_{\sep(\classic{C}{:}\classic{D})}^{n}
\left(\varrho\|\sepstates(\system{C}{:}\system{D})\right). \]
Indeed, we know from~\cite[Theorem 2]{Piani} 
that $D_{\sep}(\cdot\|\sepstates)$ is a super-additive quantity, 
which means that, for any $n\in\N$, 
$D_{\sep}\left(\gamma^{\otimes n}\|\sepstates\right)/n 
\geq D_{\sep}(\gamma\|\sepstates)$.

Let us briefly expand on why we claim that upper bounding 
$D_{\sep}^{\infty}\left(\gamma\|\sepstates\right)$ 
as tightly as possible would be of interest. 
First, we know from~\cite{BCHW} that the quantum key repeater rate of two copies of $\gamma$ (without the $\Rcopies$-bounded restriction) is upper bounded by 
$D_{\locc}^{\infty}\left(\gamma\|\sepstates\right)$ 
(an upper bound that we slightly improve in Appendix \ref{appendix:R_D}), 
hence upper bounded by $D_{\sep}^{\infty}\left(\gamma\|\sepstates\right)$. 
Second, we know from~\cite{LiW} that the distillable entanglement 
is upper bounded by $D_{\locc}^{\infty}\left(\gamma\|\sepstates\right)$, 
hence upper bounded by $D_{\sep}^{\infty}\left(\gamma\|\sepstates\right)$. 
Any upper bound on $D_{\sep}^{\infty}\left(\gamma\|\sepstates\right)$ is thus 
automatically an upper bound on these two important operational quantities. 
Unfortunately, we are not able to prove with our current techniques that 
$D_{\sep}^{\infty}\left(\gamma\|\sepstates\right)$ 
is with high probability small. 

Related to this comment, let us emphasize once more that the ``usual'' bounds on quantities such as the quantum key repeater rate or the distillable entanglement, based on the partial transposition, are not useful in the present case. For instance, the log negativity of our random private state $\gamma$, i.e.~$E_N(\gamma)\coloneqq \log\norm{\gamma^{\Gamma}}_1$, is with high probability high. Indeed, in matrix notation we have
\[ \gamma^{\Gamma} = \frac{1}{2} \begin{pmatrix} \openone/d^2 & 0 & 0 & 0 \\ 0 & 0 & {(\randomshieldplus-\randomshieldminus)}^{\Gamma}/2 & 0 \\ 0 & {(\randomshieldplus-\randomshieldminus)}^{\Gamma}/2 & 0 & 0 \\  0 & 0 & 0 & \openone/d^2 \end{pmatrix}, \]
and we thus easily see that 
$\norm{\gamma^{\Gamma}}_1= 1 + 
\norm{{(\randomshieldplus-\randomshieldminus)}^{\Gamma}}_1$. 
Now, the spectrum of the random matrix ${(\randomshieldplus-\randomshieldminus)}^{\Gamma}$ can be precisely studied (see e.g.~\cite[Section 3]{Montanaro}), but for our purposes it is in fact enough to simply know that there exists a universal constant $\lowerconstant>0$ such that $\norm{{(\randomshieldplus-\randomshieldminus)}^{\Gamma}}_1\geq \lowerconstant$ with high probability. And therefore, $E_N(\gamma)\geq \log(1+\lowerconstant)$ with high probability.
\smallskip

For the sake of clarity, we focused in \Cref{sec:p-bit} on one particular way of constructing random private states. However, the properties that we described would hold true for many other random private state models. For instance, one could think of picking as states $\randomshieldplusminus$, two independent uniformly distributed mixed states on $\system{A}'\system{B}'$, or mixtures of order $d^2$ independent uniformly distributed pure states on $\system{A}'\system{B}'$. These would be with high probability approximately orthogonal, so that the random state $\gamma$ on $\system{A}\system{A}'\system{B}\system{B}'$ formed out of them would be with high probability an approximate private state. Moreover, it would have with high probability all the previously observed features. It thus appears as a generic aspect of private states that their amount of distillable entanglement and their amount of data-hiding have to obey some trade-off. One important open question at this point would nonetheless be: what is the actual distribution of the random private states which are produced in ``usual'' quantum key distribution protocols? Indeed, however wide the range of models our results apply to, it would be interesting to know whether or not the outputs of error correction and privacy amplification procedures which are performed in practice fall into this general framework.

\section{Acknowledgements}

  We would like to thank Ludovico Lami and an anonymous referee for suggesting improvements to the presented results. This research was financially supported by the European Research Council (Grant agreements No.~337603 and No.~648913 and No.~818761), the Villum Centre of Excellence for the Mathematics of Quantum Theory, the Danish Council for Independent Research (Sapere Aude), the John Templeton Foundation (Grant No.~48322), the French National Centre for Scientific Research (ANR Project Stoq 14-CE25-0033) and the Bundesministerium f\"ur Bildung und Forschung (BMBF) through Grant 16KIS0857.

\appendices

\section{Local distinguishability of isotropic states}
\label{section:iso}

    \newcommand{\iso}{{iso}}
    \newcommand{\isostate}{\iota}
    \newcommand{\isooperator}{I}
    \newcommand{\local}{{\mathbf{LO}}}
    \newcommand{\localchannels}{\local}
    \newcommand{\localmeasurements}{\local}
    \newcommand{\isotwirl}{\mathcal{I}}

    Fix a product basis $\quantity{\ket{ij}}$ on $\system{C}\system{D}$ and define the following states and projectors:
    \begin{align*}
    P      &\coloneqq\frac{1}{d}\sum_{ij}\ketbra{ii}{jj} &
    P_\perp&\coloneqq\openone-P
    \\
    \psi      &\coloneqq\frac{P}{\Tr P}=P & \psi_\perp&\coloneqq\frac{P_\perp}{\Tr P_\perp}
    \end{align*}
    Let $p\in[0,1]$, a general isotropic state has the form
    \begin{align*}
    \isostate(p) &\coloneqq p \psi + (1-p) \psi_\perp \,,
    \end{align*}
    The action of the isotropic twirl $\isotwirl$ 
    on any Hermitian operator $X$ on $\system{A}\system{B}$ is
    \begin{align*}
    \isotwirl(X) &= (\Tr XP)\psi + (\Tr XP_\perp)\psi_\perp \;.
    \intertext
    {It is important to note that the isotropic twirl 
    can be implemented using one-way LOCC by sampling uniformly at random 
    a finite set of product unitaries~\cite{DCEL}, thus it can be used to construct $\locc$, $\sep$ and $\ppt$ operations, 
    as is the case of \Cref{lemma:pbit-M-norms-lower}.
    With some abuse of notation, let us denote ${U}\in\isotwirl$ unitaries sampled in the twirl, and with $|\isotwirl|$ the finite number of such unitaries in the twirl, then}
    \isotwirl(X) &=  \frac{1}{|\isotwirl|} \sum_{{U}\in\isotwirl}
    (\mathcal{U}\otimes\overline{\mathcal{U}})(X)\;.
    \end{align*}
    where $\mathcal{U}(\rho)= U \rho U^\dagger$ and $\overline{\mathcal{U}}(\rho)= \overline U \rho U^T$ are the conjugations by $U$ and $\overline U$, and $\overline U$ denotes the complex conjugate (which is basis dependent).

\begin{lemma}[\cite{LiW}]
\label{lemma:iso:distinguishability}
Let $\eta(x\|y)\coloneqq x(\log x - \log y)$,  $\sepstates\equiv\sepstates(\system{C}{:}\system{D})$, 
$\localmeasurements\equiv\localmeasurements(\classic{C}{:}\classic{D})$ and 
$\pptmeasurements  \equiv\pptmeasurements  (\classic{C}{:}\classic{D})$.
For any isotropic states
$\rho   = \isostate(p)$ and $\sigma = \isostate(q)$
on $\system{C}\system{D}$ we have
\begin{align*}
\norm{\rho-\sigma}_\localchannels&=
\norm{\rho-\sigma}_\pptchannels   =
    2\frac{d}{d+1}|p-q|
\\
D_\localmeasurements(\rho\|\sigma)   &=
D_\pptmeasurements  (\rho\|\sigma) \\&=
    \frac{d}{d+1} \quantity[ \vphantom{A_p^T} 
    \eta\quantity( p+\frac{1}{d}\middle\|q+\frac{1}{d})+
    \eta\quantity( 1-p\middle\|1-q) ]
\intertext{and consequently if $\rho$ is entangled ($p\geq1/d$)}
\norm{\rho-\sepstates}_\localchannels&=
\norm{\rho-\sepstates}_\pptchannels   =
    2\frac{d}{d+1}\quantity(p-\frac{1}{d})\;.
\\
D_\localmeasurements(\rho\|\sepstates)   &=
D_\pptmeasurements  (\rho\|\sepstates) \\&=
    \frac{d}{d+1} \quantity[ \vphantom{A_p^T} 
    \eta\quantity( p+\frac{1}{d}\middle\|\frac{2}{d})+
    \eta\quantity( 1-p\middle\|1-\frac{1}{d}) ]
\,.
\end{align*}
\end{lemma}
\begin{proof}
The proof follows step by step the proof for the local relative entropy of the maximally entangled state found in~\cite[Proposition~4]{LiW}.
First we estimate the lower bounds using 
the (local) measurement in the computational basis. 
For the purpose, consider the following parametrization of the isotropic states:
\begin{align*}
\rho &= p\psi + (1-p)\psi_\perp = a\psi + (1-a)\frac{\openone}{d^2}
\\
\sigma &= q\psi + (1-q)\psi_\perp = b\psi + (1-b)\frac{\openone}{d^2}
\end{align*}
which gives
\begin{align*}
p&=a\frac{d^2-1}{d^2} + \frac{1}{d^2}
\\
q&=b\frac{d^2-1}{d^2} + \frac{1}{d^2}
\\
|p-q|&=\frac{d^2-1}{d^2}|a-b|
\;.
\end{align*}
For any $\channels$ norm we have
\begin{align*}
\norm{\rho-\sigma}_\channels
&= |p-q|\cdot\norm{\psi-\psi_\perp}_\channels \\
&= |a-b|\cdot\norm{\psi-\tfrac{\openone}{d^2}}_\channels
\;,
\end{align*}
meaning that the optimal measurement is independent of the isotropic states.
In particular we obtain
\[\frac{d^2-1}{d^2}\norm{\psi-\psi_\perp}_\channels
= \norm{\psi-\tfrac{\openone}{d^2}}_\channels\;.\]
The measurement in the computational basis then gives
\begin{align}
\norm{\psi-\frac{\openone}{d^2}}_\localchannels\geq 2\frac{d-1}{d}
\end{align}
and thus
\begin{align}
\label{eq:norm:iso:lower}
\norm{\rho-\sigma}_\localchannels
=    2\frac{d}{d+1}|p-q|\;.
\end{align}
For the relative entropy the measurement yields 
\begin{align}
D_\localmeasurements&(\rho\|\sigma)
    \nonumber
\\&\geq 
    \sum_{ij} \eta \left( 
    \frac{a}{d}\delta_{ij}+\frac{1-a}{d^2} \middle\|
    \frac{b}{d}\delta_{ij}+\frac{1-b}{d^2} \right)
    \nonumber
\\&\geq
    \frac{d-1}{d} \quantity( \vphantom{A_p^T} 
    \eta\quantity\big( a+\tfrac{1}{d-1}\big\|b+\tfrac{1}{d-1})+
    \eta\quantity( 1-a\middle\|1-b) )
    \nonumber
\\&=
    \frac{d}{d+1} \quantity( \vphantom{A_p^T} 
    \eta\quantity( p+\tfrac{1}{d}\middle\|q+\tfrac{1}{d})+
    \eta\quantity( 1-p\middle\|1-q) )
\label{eq:relent:isoiso:lower}
\end{align}
Notice that in both cases, the outcome of the local measurement is the same as the outcome of the binary projective measurement on the maximally correlated subspace and the remaining orthogonal subspace.

To upper bound $\norm{\rho-\sigma}_\pptchannels$, 
we use that any measurement acting on $\rho$ and $\sigma$ can be reduced to an isotropic measurement using $\Tr M\isostate(p)=\Tr \isotwirl(M)\isostate(p)$, where $M$ is any positive operator~\cite{HH}.
If $M$ is a PPT operator, then $\isotwirl(M)$ will be a PPT isotropic operator, all a which can be decomposed into a combination of the two extremal PPT isotropic operators $P+P_\perp/(d+1)$ and $P_\perp$ (see \Cref{figure:iso}).
We can thus fine grain the measurement into operators proportional to the two extremal ones and then join them into an isotropic binary measurement.
The same is true for the relative entropy using first joint convexity to fine grain the measurement into the extremal operators, and then using 
$\eta(ax\|ay)+\eta(bx,by) = \eta((a+b)\|(a+b)y)$ to join them into a binary measurement.
The result is, that we can restrict to binary measurements without loss of generality, and that the optimal measurement is the binary measurement with the two extremal PPT points as measurement operators.

Let $\alpha,\beta\in[0,1]$, a general isotropic measurement operator has the form
\begin{align*}
\isooperator^{\alpha,\beta} & = \alpha P + \beta P_\perp
\intertext{and thus the associated dual operator (see \Cref{sec:Mnorms}) of an isotropic binary measurement 
$\measurement=(\isooperator^{\alpha,\beta},\openone-\isooperator^{\alpha,\beta})$
has the form:}
    K_\isooperator^{\alpha,\beta} &= (2\alpha-1) P + (2\beta-1) P_\perp
    \;.
\end{align*}
{We thus find that:}
\begin{align}
\label{eq:norm:isoiso}
\norm{\measurement
\quantity\big(\psi-\psi_\perp)}_1
&=\Tr[K_\isooperator^{\alpha,\beta} (\psi-\psi_\perp)]
=2\cdot|\alpha -\beta|
\;.
\end{align}
The extremal operators are at
$\alpha=1$ and $\beta=\frac{1}{d+1}$, giving
\[\norm{\psi-\psi_\perp}_\pptchannels \leq 2\frac{d}{d+1}\]
and matching the lower bound of \Cref{eq:norm:iso:lower}. This proves
\begin{align*}
\norm{\rho-\sigma}_\localchannels =
\norm{\rho-\sigma}_\pptchannels   = 2\frac{d}{d+1}|p-q|\;.
\end{align*}
Similarly for the relative entropy we find that the extremal measurement achieves
\begin{align*}
D&_\pptmeasurements(\rho\|\sigma)  
\\&\leq
    \eta\quantity(p+\frac{1-p}{d+1}\middle\|q+\frac{1-q}{d+1})
   +\eta\quantity(\frac{(1-p)d}{d+1}\middle\|\frac{(1-q)d}{d+1})
\\&=\frac{1}{d+1}\eta\quantity(pd+1\middle\|qd+1)
   +\frac{d}{d+1}\eta\quantity(1-p\middle\|1-q)
\\&=\frac{d}{d+1} \quantity( 
    \eta\quantity( p+\frac{1}{d}\middle\|q+\frac{1}{d})+
    \eta\quantity( 1-p\middle\|1-q) )
\end{align*}
matching the lower bound in \Cref{eq:relent:isoiso:lower} and proving
\begin{align}
D&_\localmeasurements(\rho\|\sigma)
=D_\pptmeasurements  (\rho\|\sigma)
\nonumber
\\&=\frac{d}{d+1} \quantity( 
    \eta\quantity( p+\frac{1}{d}\middle\|q+\frac{1}{d})+
    \eta\quantity( 1-p\middle\|1-q) )\;.
\label{eq:relent:isoiso}
\end{align}

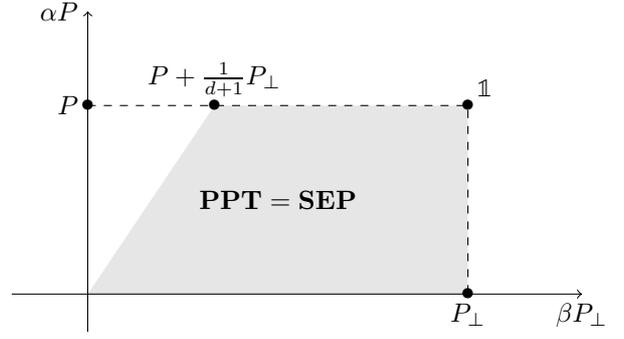
\begin{figure}
\centering
\begin{tikzpicture}[xscale=5,yscale=2.5]
    \draw [->] (-.2,0) -- (1.3,0) node [below] {$\beta P_\perp$};
    \draw [->] (0,-.2) -- (0,1.5) node [left ] {$\alpha P$};

    \draw (1,1) node {\textbullet} node [above right] {$\openone$};
    \draw (1/3,1) node {\textbullet} node [above] {$P+\frac{1}{d+1}P_\perp$};
    \draw (0,1) node {\textbullet} node [left ] {$P$};
    \draw (1,0) node {\textbullet} node [below] {$P_\perp$};
    
    \node at (.5,.5) {$\ppt=\sep$};

    \draw [dashed] (0,1) -- (1,1) -- (1,0);
    \fill [opacity=.1,black] (0,0) -- (1/3,1) -- (1,1) -- (1,0) -- cycle;
    
\end{tikzpicture}
\caption{\label{figure:iso}The space of isotropic measurement operators, the white area are the ``entangled'' measurement operators. Notice that any separable operator can be written as linear combination of $P_\perp$ and $P+\frac{1}{d+1}P_\perp$.}
\end{figure}

For the second part of the claim, we need to use the isotropic twirl to argue that it is enough to look at isotropic separable states to compute $\norm{\rho-\sepstates}_\localchannels$ and $D_\localchannels(\varrho\|\sepstates)$. 
However, as a channel, the isotropic twirl needs shared randomness/communication and thus is not in $\localchannels$,
therefore we need to use 
the convexity of the $\localmeasurements$ norm and 
the joint convexity of the $\localmeasurements$ relative entropy 
to de-randomize it.
Namely, let $\varsigma\in\sepstates$ be such that 
$\norm{\varrho-\sepstates}_\localchannels = \norm{\varrho - \varsigma}_\localchannels$, then
\begin{align*}
\norm{\rho -\isotwirl(\varsigma)}_\localchannels
  &=\norm{\isotwirl(\rho) -\isotwirl(\varsigma)}_\localchannels
\\&\leq \frac{1}{|\isotwirl|} 
    \sum_{{U}\in\isotwirl}
    \norm{(\mathcal{U}\otimes\overline{\mathcal{U}})(\rho-\varsigma)}_\localchannels
\intertext
{Since $\mathcal{U}\otimes\overline{\mathcal{U}}$ is a reversible local operation, we now use that $\norm{(\mathcal{U}\otimes\overline{\mathcal{U}})(X)}_\localchannels
=\norm{X}_\localchannels$, and get}
\norm{\rho -\isotwirl(\varsigma)}_\localchannels
  &= \frac{1}{|\isotwirl|} 
    \sum_{{U}\in\isotwirl}
    \norm{\rho-\varsigma}_\localchannels
\\&=\norm{\rho-\varsigma}_\localchannels
\\&=\norm{\rho-\sepstates}_\localchannels
\\&\leq
    \norm{\rho-\sepstates}_\pptchannels
\\&\leq
    \norm{\rho-\isostate(\tfrac{1}{d})}_\pptchannels
\;.
\end{align*}
Similarly, if  $\varsigma\in\sepstates$ is such that 
$D_\localchannels(\varrho\|\sepstates) = D_\localchannels(\varrho\|\varsigma)$,
then using the joint convexity of $D_\localchannels$ we get 
\begin{align*}
D_\localchannels(\rho\|\isotwirl(\varsigma))
  &\leq D_\localchannels(\rho\|\varsigma)
\\& =   D_\localchannels(\rho\|\sepstates)
\\&\leq D_\pptchannels(\rho\|\sepstates)
\\&\leq D_\pptchannels(\rho\|\isostate(\tfrac{1}{d}))
\;.
\end{align*}
However it is straightforward to check that over the separable isotropic states, \Cref{eq:norm:isoiso,eq:relent:isoiso} achieve the minimum for $\sigma=\isostate(\frac{1}{d})$ and therefore we find
\begin{align*}
\norm{\rho -\isostate(\tfrac{1}{d})}_\localchannels &\leq
\norm{\rho -\isotwirl(\varsigma)   }_\localchannels
\\
D_\localchannels(\rho\|\isostate(\tfrac{1}{d})) &\leq
D_\localchannels(\rho\|\isotwirl(\varsigma))
\end{align*}
concluding the proof.
\end{proof}

\section{Local distinguishability of random private quantum states from separable states}
\label{appendix:improvement}

\subsection{Operator ordering with respect to measurements}
\label{sec:data-hiding}

Contrary to the definitions of $\measurements$ norm and $\measurements$ relative entropies (see \Cref{sec:Mnorms}), the definition below, as far as we are aware of, has not been introduced in the literature before.

\begin{definition}[$\mathbf{M}$ (partial) ordering]
    \label{def:M-order}
    For any Hermitian operators $X,Y$ on $\system{H}$, we define the notion of ordering in restriction to $\mathbf{M}$ by:
    \[ X\leq_{\mathbf{M}}Y\ \ \text{if}\ \ \forall\ \mathcal{M}\in\mathbf{M},\ \mathcal{M}(X) \leq  \mathcal{M}(Y)
    \,.\]
\end{definition}

Note that the condition in \Cref{def:M-order} above can be rewritten as point-wise ordering, i.e.~writing $\mathcal{M}=(T_i)_{i\in I}$,
\[\mathcal{M}(X) \leq \mathcal{M}(Y) 
\ \ \text{if}\ \ \forall\ i\in I,\ \Tr(T_i X)\leq\Tr(T_i Y) \,.\]

We now explore how this notion of measurement ordering connects to that of measurement distance. We begin with the following easy observations in \Cref{lemma:order-relent,lemma:order-trace} below. \Cref{lemma:order-relent} will be used later in the paper, while \Cref{lemma:order-trace} is just stated here as an independent comment.

\begin{lemma}[Relating $\mathbf{M}$ ordering and $\mathbf{M}$ relative entropy distance]
    \label{lemma:order-relent}
    Let $\varrho,\sigma$ be states and $\measurements$ be a set of measurements on $\system{H}$.
    If \(\varrho \leq _{\measurements} (1+\epsilon)\sigma \) for some $\epsilon>0$, then \(D_{\measurements} (\varrho\parallel\sigma) \leq  \log (1 + \epsilon) \).
\end{lemma}

\begin{proof}
    If $p,q$ are probability distributions satisfying $p\leq (1+\epsilon)q$ for some $\epsilon>0$, then clearly 
    \begin{align*}
    D(p\parallel q) 
      &= \sum_i p_i\log\left(\frac{p_i}{q_i}\right) 
    \\&\leq \sum_i p_i\log(1+\epsilon) 
    \\&= \log(1+\epsilon)
    \, .
    \end{align*}
    Now, for any \(\measurement\in \measurements\),
    \(\measurement (\varrho)\) and \(\measurement (\sigma)\) are classical probability distributions. So what we have shown is that, if \(\measurement(\varrho) \leq(1+\epsilon)\measurement(\sigma)\) for all \(\measurement \in \measurements\), then \(D\left(\measurement(\varrho)\|\measurement(\sigma)\right)\leq \log\left(1+\varepsilon\right)\) for all \(\measurement \in \measurements\). And this is exactly the statement in \Cref{lemma:order-relent}. 
\end{proof}

\begin{lemma}[Relating $\mathbf{M}$ ordering and $\mathbf{M}$ norm distance]
    \label{lemma:order-trace}
    Let $\varrho,\sigma$ be states and $\measurements$ be a set of measurements on $\system{H}$.
    If \(\varrho \leq _{\measurements} (1+\epsilon)\sigma \) and \(\sigma \leq _{\measurements} (1+\epsilon)\varrho \) for some $0<\epsilon<1$, then $\|\varrho-\sigma\|_{\measurements}\leq\epsilon/(1-\epsilon/2) \leq 2\epsilon$.
\end{lemma}

\begin{proof}
    If $p,q$ are probability distributions satisfying $p\leq (1+\epsilon)q$ and $q\leq(1+\epsilon)p$ for some $0<\epsilon<1$, then
    \[ \forall\ i,\ p_i-q_i\leq \epsilon q_i\ \text{and}\ q_i-p_i\leq \epsilon p_i. \]
    Hence as a consequence,
    \begin{align*} 
    \sum_i|p_i-q_i|
      &\leq\epsilon\sum_i\max(p_i,q_i)
    \\&= \epsilon\sum_i\frac{p_i+q_i+|p_i-q_i|}{2}
    \\&=\epsilon \left(1+\frac{1}{2}\sum_i|p_i-q_i|\right).
    \end{align*}
    Now, for any \(\measurement\in \measurements\),
    \(\measurement (\varrho)\) and \(\measurement (\sigma)\) are classical probability distributions. So what we have shown is that, if \(\measurement(\varrho) \leq(1+\epsilon)\measurement(\sigma)\) and \(\measurement(\sigma) \leq(1+\epsilon)\measurement(\varrho)\) for all \(\measurement \in \measurements\), then \(\|\measurement(\varrho)-\measurement(\sigma)\|_1\leq\epsilon/(1-\epsilon/2)\) for all \(\measurement \in \measurements\). And this is exactly the statement in \Cref{lemma:order-trace}. 
\end{proof}

\subsection{SEP operator ordering for random private quantum states}

We start with establishing a technical result about the maximum overlap with separable states for the difference of two random orthogonal states. It has some similarities with the SEP data hiding result of~\cite[Theorem 5]{AL}, but cannot be directly derived from it, which is why we re-do the whole argument.

\begin{proposition}
    \label{prop:data-hiding}
    Let $\randomshieldplusminus$ be random orthogonal states on ${\system{A}'\system{B}'}$ as defined by \Cref{con:random-states}. Then, there exist universal constants $\probconstant,\upperconstant>0$ such that
    \[ \probability(
    \sup_{\tau\in\sepstates(\system{A}'{:}\system{B}')} 
    \quantity|\Tr\quantity\big(\tau [\randomshieldplus-\randomshieldminus])|
    \leq \frac{\upperconstant}{d^{5/2}} ) 
    \geq 1-e^{-{\probconstant}d}\, . 
    \]
\end{proposition}

\begin{proof}
    With some abuse of notation, let us define the following function
    on the traceless Hermitian operators
    \[\norm{X}_{\sepstates^\circ} \coloneqq \sup_{\tau\in\sepstates(\system{A}'{:}\system{B}')} 
        \quantity\big|\Tr\quantity(\tau X)|\]
    and let us remark that it satisfies the triangle inequality.
    To be more precise, $\norm{X}_{\sepstates^\circ}$ is the support function of the symmetrization of the separable state $\Sigma = \conv\quantity{-\sepstates\cup\sepstates}$, and thus $\norm{X}_{\sepstates^\circ}=\norm{X}_{\Sigma^\circ}$, namely it is the gauge of the polar of $\Sigma$, see \cite{AS} for more details.
    Notice that $\norm{X}_{\sepstates^\circ}$ is not to be confused with $\norm{X}_{\sepstates_0}$ introduced in \Cref{lemma:robustness:random}.
    In particular the former is always smaller than the $\infty$-norm, while the latter is always larger than the trace norm. Still, we will follow the same proof structure of \Cref{lemma:robustness:random}.
    
    Recall that the random states are defined by 
    $\randomshieldplusminus = U \fixedshieldplusminus U^\dagger$, where $U$ is a Haar-distributed unitary on $\C^d\otimes\C^d$
    and $\fixedshieldplusminus$ some fixed orthogonal maximally mixed states  on $d^2/2$-dimensional subspaces of $\C^d\otimes\C^d$.
    Let
    \begin{align*}
    \bar\Delta
    &=\fixedshieldplus-\fixedshieldminus
    \\
        \Delta
    &=                \randomshieldplus-                \randomshieldminus
    =U\bar\Delta U^\dagger
    \;.
    \end{align*}
    The statement to prove therefore is
    \[ \probability(
    \norm{\Delta}_{\sepstates^\circ} 
    \leq \frac{\upperconstant}{d^{5/2}} ) 
    \geq 1-e^{-{\probconstant}d}. \]
    To prove the statement, we compute the expectation value $\E\norm{\Delta}_{\sepstates^\circ}=\E\norm{U\bar\Delta U^\dagger}_{\sepstates^\circ}$ over the random variable $U$, then we estimate the Lipschitz constant of $\norm{U\bar\Delta U^\dagger}_{\sepstates^\circ}$ as a function of $U$ and use this to argue that being close to the expected value happens with high probability. $\norm{X}_{\sepstates^\circ}$ is not unitary invariant, however the function $\E\norm{UXU^\dagger}_{\sepstates^\circ}$ is unitary invariant on $X$, while still being convex.
    
    As explained in \Cref{lemma:robustness:random} leading to \Cref{eq:traceless:convex:twice}, we know from~\cite[Lemma~6]{AL} that for any unitary-invariant convex function $g$ of any traceless Hermitian operators $X$ and $Y$ on $\C^d\otimes\C^d$, we have
    \begin{equation*}
    \frac{1}{2d^2} \frac{\norm{X}_1}{\norm{Y}_\infty}
    \leq \frac{g(X)}{g(Y)}
    \leq 2d^2 \frac{\norm{X}_\infty}{\norm{Y}_1}
    \;.
    \end{equation*}
    Now, we let $Y=\bar\Delta$ for which $\norm{\bar\Delta}_1=2$ and $\norm{\bar\Delta}_\infty = 2/d^2$:
    \begin{equation*}
    \frac{1}{4} {\norm{X}_1} 
    \leq \frac{g(X)}{g(\bar\Delta)}
    \leq d^2 {\norm{X}_\infty}
    \;.
    \end{equation*}
    Then we let $g(X)=\E\norm{UXU^\dagger}_{\sepstates^\circ}$, which is unitary invariant by construction and convex by the convexity of $\norm{X}_{\sepstates^\circ}$:
    \begin{equation*}
    \frac{1}{4} {\norm{X}_1} 
    \leq \frac{\E\norm{UXU^\dagger}_{\sepstates^\circ}}{\E\norm{U\bar\Delta U^\dagger}_{\sepstates^\circ}}
    \leq d^2 {\norm{X}_\infty}
    \;.
    \end{equation*}
    We now let $X$ be again a Gaussian vector $\GUE$ on the traceless Hermitian operators (Gaussian unitary ensemble) on $\C^d\otimes\C^d$.
    This makes 
    $\E\norm{U{\GUE}U^\dagger}_{\sepstates^\circ}
    =\E\norm{\GUE}_{\sepstates^\circ}$.
    Like in \Cref{lemma:robustness:random}, taking expectation values over the inequalities gives
    \begin{equation*}
    \frac{1}{4} \E{\norm{\GUE}_1} 
    \leq \frac{\E\norm{\GUE}_{\sepstates^\circ}}
              {\E\norm{\Delta}_{\sepstates^\circ}}
    \leq d^2 \E{\norm{\GUE}_\infty}
    \;,
    \end{equation*}
    and using that $\E\norm{\GUE}_1\sim d^3$ and $\E\norm{\GUE}_{\infty}\sim d$ further gives:
    \begin{equation*}
    \E\norm{\Delta}_{\sepstates^\circ} \sim
     \frac{1}{d^3} \E\norm{G}_{\sepstates^\circ} 
    \;.
    \end{equation*}
    We now know from~\cite[Equation~7]{AS}
    \footnote{As remarked in~\cite[Section~2.1]{AS}, for a convex set $K$ of $\R^n$ we have
    $\E\norm{G}_{K^\circ} = \gamma \,w(K) \geq \gamma \vrad(K)$, where $\gamma=\E\norm{G}_2$, $w$ is the mean width and $\vrad$ is the volume radius. Furthermore, \cite{AS} shows that in particular $w(\Sigma)\sim\vrad(\Sigma)$. What~\cite[Equation~7]{AS} and the remarks below show, is that we have $\vrad(\Sigma)\sim \sqrt d / \gamma$, which gives us $\E\norm{G}_{\sepstates^\circ}\sim\sqrt{d}$.}
    that $\E\norm{G}_{\sepstates^\circ}$ is at most of order $\sqrt{d}$ and,
    therefore there exists a universal constant $\upperconstant>0$ such that 
    \begin{equation}
    \label{eq:data-hiding-expectation}
    \E\norm{\Delta}_{\sepstates^\circ} \leq \frac{\upperconstant}{d^{5/2}}.
    \end{equation}
    
    Now we have to show that this average behaviour is generic for large $d$, because the function $f(U)\coloneqq\norm{U\bar\Delta U^\dagger}_{\sepstates^\circ}$ is regular enough in the Euclidean norm: we claim that it is a $8/d^2$-Lipschitz function. Indeed by the triangle inequality for $\norm{\,\cdot\,}_{\sepstates^\circ}$, for any unitaries $U$ and $V$ on $\C^d\otimes\C^d$ we have 
    \begin{align*}
   |{f}(U)&-{f}(V)|
    \\&=
    \quantity\Big
    |\norm{U\bar\Delta U^\dagger}_{\sepstates^\circ} -
     \norm{V\bar\Delta V^\dagger}_{\sepstates^\circ} |
    \\&\leq 
        \norm{U\bar\Delta U^\dagger - V\bar\Delta V^\dagger}_{\sepstates^\circ}
    \\&\leq
        \norm{U \fixedshieldplus  U^\dagger-V \fixedshieldplus V^\dagger}_{\sepstates^\circ} +
        \norm{U \fixedshieldminus U^\dagger-V \fixedshieldminus V^\dagger}_{\sepstates^\circ}
    \,.
    \end{align*}
    {We can then use that $\sepstates\subset B_1$, together with duality 
    of the $1$-norm and the $\infty$-norm, to get}
    \begin{align*}
    |{f}(U)&-{f}(V)|
    \\&\leq
    \|U \fixedshieldplus U^\dagger-V \fixedshieldplus V^\dagger\|_{B_1^\circ} + \|U \fixedshieldminus U^\dagger-V \fixedshieldminus V^\dagger\|_{B_1^\circ} 
    \\&= 
    \|U \fixedshieldplus U^\dagger-V \fixedshieldplus V^\dagger\|_{\infty} + \|U \fixedshieldminus U^\dagger-V \fixedshieldminus V^\dagger\|_{\infty} 
    \\&\leq 
    \|U \fixedshieldplus (U^\dagger-V^\dagger)\|_\infty +
    \|(U-V) \fixedshieldplus V^\dagger\|_{\infty} 
    \\& \quad+
    \|U \fixedshieldminus (U^\dagger-V^\dagger)\|_\infty +
    \|(U-V) \fixedshieldminus V^\dagger\|_{\infty} 
    \\ & =
    2\|(U-V) \fixedshieldplus\|_{\infty} +
    2\|(U-V) \fixedshieldminus\|_{\infty} 
    \\ & \leq
    2 \|\fixedshieldplus\|_\infty \|U-V\|_\infty +
    2 \|\fixedshieldminus\|_\infty \|U-V\|_\infty 
    \\
    & = \frac{8}{d^2} \norm{U-V}_{\infty}
    \leq \frac{8}{d^2}\norm{U-V}_{2}. 
    \end{align*}
    which shows that $f$ is $8/d^2$-Lipschitz.
    
    Now, we know from~\cite[Corollary 17]{MM} that any $L$-Lipschitz function $g$ on the unitaries on $\C^D$ (equipped with the Euclidean metric) satisfies the concentration estimate: if $U$ is a Haar-distributed unitary on $\C^D$, then for all $\epsilon>0$, $\probability( g(U) > \E g + \epsilon ) \leq e^{-{\probconstant}D \epsilon^2/L^2}$, where $\probconstant>0$ is a universal constant. Combining the above estimate on the Lipschitz constant of $\norm{U\bar\Delta U^\dagger}_{\sepstates^\circ}$ with the estimate on its expected value from \Cref{eq:data-hiding-expectation}, we thus get for all $\epsilon>0$
    \begin{align*}
    \probability(\norm{\Delta}_{\sepstates^\circ} > \frac{\upperconstant}{d^{5/2}}+\epsilon) 
      &\leq
    \probability\big(\norm{\Delta}_{\sepstates^\circ} > \E\norm{\Delta}_{\sepstates^\circ}+\epsilon) 
    \\&\leq 
    e^{-\probconstant d^6\epsilon^2/64}.
    \end{align*}
    The advertised result follows from choosing $\epsilon={\upperconstant}/d^{5/2}$ (and suitably relabelling the constants).
\end{proof}

Thanks to \Cref{prop:data-hiding}, we can now show that our random private states and their key-attacked versions are with high probability SEP ordered with a constant close to $1$.

\begin{proposition}
    \label{prop:gamma-tau-SEP-order}
    Let $\gamma$ and $\hat{\gamma}$ on ${\system{A}\system{A}'\system{B}\system{B}'}$ 
    be a random private state and its key-attacked state as defined by \Cref{eq:def-random-gamma}. Then, there exist universal constants $\probconstant,\upperconstant>0$ such that
    \[ \probability(\gamma \leq_{\sep(\classic{A}\classic{A}'{:}\classic{B}\classic{B}')} \left( 1+\frac{\upperconstant}{\sqrt{d}} \right)\hat{\gamma} ) \geq 1-e^{-{\probconstant}d}\, . \]
\end{proposition}

\begin{proof}
    By \Cref{def:M-order}, to prove \Cref{prop:gamma-tau-SEP-order} it suffices to show that, with probability greater than $1-e^{-{\probconstant}d}$,
    forall positive operators $0\leq M,N\leq \openone$
    \begin{equation} 
    \label{eq:M-N} 
    \Tr\quantity(M\otimes N \gamma)
    \leq \quantity(1+\frac{\upperconstant}{\sqrt{d}}) 
    \Tr\quantity(M\otimes N \hat{\gamma}).
    \end{equation}
    Now, let $\Delta = \randomshieldplus-\randomshieldminus$ and observe that, for any $\epsilon>0$, we have 
    \[ (1+\epsilon)\hat{\gamma}-\gamma = \frac{1}{2} 
    \begin{pmatrix}
    \displaystyle\epsilon\,\frac{\openone}{d^2} & 0 & 0 & \displaystyle -\frac{\Delta}{2} 
    \\ 0 & 0 & 0 & 0 
    \\ 0 & 0 & 0 & 0 
    \\ \displaystyle-\frac{\Delta}{2} 
    & 0 & 0 & \displaystyle \epsilon\,\frac{\openone}{d^2} \end{pmatrix}.\]
    
    Given $0\leq M\leq \openone$ on $\system{A}\system{A}'$ and $0\leq N\leq \openone$ on $\system{B}\system{B}'$, we write them in the block-form
    \[ M=\begin{pmatrix} M_1 & \hat{M} \\ \hat{M}^{\dagger} & M_2 \end{pmatrix} \ \text{and}\ N=\begin{pmatrix} N_1 & \hat{N} \\ \hat{N}^{\dagger} & N_2 \end{pmatrix}, \]
    where $M_1,M_2,\hat{M}$ are operators on $\system{A}'$ and  $N_1,N_2,\hat{N}$ are operators on $\system{B}'$, with $0\leq M_1,M_2\leq \openone$ and $0\leq N_1,N_2\leq \openone$. 
    A straightforward calculation shows that
    \begin{equation}
    \label{eq:delta-prime}
    \Tr\left(M\otimes N\quantity\big[\left(1+\epsilon\right)\hat{\gamma}-\gamma]\right)
    = \epsilon\frac{\Tr S}{d^2} -\frac{1}{2}\Tr(\hat{S}{\Delta})
    \end{equation}
    where ${S}=\quantity\big(M_1\otimes N_1+M_2\otimes N_2)/2 \geq 0$ and $\hat{S}=\hat{S}^\dagger 
    = \quantity\big(\hat{M}\otimes\hat{N}+\hat{M}^{\dagger}\otimes \hat{N}^{\dagger})/2$,
    satisfying $\pm\hat{S} \leq {S}$.
    
    We now expand $\hat{S}$ into the difference 
    of the positive components of $\hat{M}$ and $\hat{N}$. 
    Namely, let us denote by $X_r$ and $X_i$ the Hermitian 
    and anti-Hermitian parts of an operator $X$, so that $X = X_r + i X_i$.
    Then we find that $\hat{S} = \hat{M}_r \otimes \hat{N}_r - \hat{M}_i \otimes \hat{N}_i$. 
    Let us further denote by $X^\pm$ the positive and negative part of an Hermitian operator $X$, so that $X=X^+-X^-$ and $|X|=X^+ + X^-$.
    Then we find that 
    $\hat{S} = \hat{S}_p-\hat{S}_n$,
    {where}
    \begin{align*}
    \hat{S}_p &
    = \hat{M}_r^+ \otimes \hat{N}_r^+ 
    + \hat{M}_r^- \otimes \hat{N}_r^- 
    + \hat{M}_i^+ \otimes \hat{N}_i^- 
    + \hat{M}_i^- \otimes \hat{N}_i^+
    \\
    \hat{S}_n &
    = \hat{M}_r^+ \otimes \hat{N}_r^- 
    + \hat{M}_r^- \otimes \hat{N}_r^+ 
    + \hat{M}_i^+ \otimes \hat{N}_i^+ 
    + \hat{M}_i^- \otimes \hat{N}_i^-
    \,.
    \end{align*}
    (but in general $|\hat{S}|\neq\hat{S}_p + \hat{S}_n$).
    By construction we have $\hat{S}_p,\hat{S}_n\in\R^+\sepstates$ and $\hat{S}_p,\hat{S}_n\leq S$, therefore
    \begin{align*}
    \quantity|\Tr(\hat{S}_p\Delta)| 
      &=    \Tr\hat{S}_p 
        \quantity|\Tr(\frac{\hat{S}_p}{\Tr\hat{S}_p}\Delta)| 
    \\&\leq \Tr\hat{S}_p \, 
        \sup_{\tau\in\sepstates} \quantity|\Tr(\tau\Delta)|
    \end{align*}
    and similarly for $\hat{S}_n$. Then
    \begin{align*}
    \Tr(\hat{S}\Delta) 
      &\leq
    \quantity|\Tr(\hat{S}_p\Delta)| + \quantity|\Tr(\hat{S}_n\Delta)|.
    \\&\leq \Tr(\hat{S}_p + \hat{S}_n) 
    \sup_{\tau\in\sepstates} \left|\Tr(\tau\Delta)\right|
    \\&\leq 2\Tr S \, \sup_{\tau\in\sepstates} \left|\Tr(\tau\Delta)\right|.
    \end{align*}
    {Now we apply \Cref{prop:data-hiding} and get that there exist constants $C,c_0>0$ such that with probability greater than $1-e^{-{\probconstant}d}$}
    \begin{align*}
    \Tr(\hat{S}\Delta) 
    &\leq 2\Tr S \, \frac{C}{\sqrt{d}} \frac{1}{d^2}.
    \end{align*}
    Inserting the above in \Cref{eq:delta-prime}
    we obtain that with probability greater than $1-e^{-{\probconstant}d}$
    we have for all $0\leq M,N\leq \openone$
    \[
    \Tr(M\otimes N\quantity\big[\left(1+\epsilon\right)\hat{\gamma}-\gamma]) 
    \geq\quantity(\epsilon-\frac{\upperconstant}{\sqrt{d}})\frac{\Tr S}{d^2}.\]
    The right-hand-side is positive as soon as $\epsilon\geq \upperconstant/\sqrt{d}$, 
    in which case \Cref{eq:M-N} indeed holds completing the proof.
\end{proof}

Notice that before this bound can be used for \Cref{th:random-R1}, the bound on the $\Rcopies$-bounded key repeater rate, it needs to be generalized to general partial measurements.

\subsection{Upper bound on the SEP relative entropy of entanglement of random private quantum states}

Using \Cref{prop:gamma-tau-SEP-order}, we are now able to prove an upper bound on $D_{\sep(\system{A}\system{A}'{:}\system{B}\system{B}')}(\gamma_{\system{A}\system{A}'\system{B}\system{B}'}
\|\sepstates(\system{A}\system{A}'{:}\system{B}\system{B}'))$ which is better than the one appearing in \Cref{th:D_SEP-D_PPT}.

\begin{theorem}
    Let $\gamma$ on ${\system{A}\system{A}'\system{B}\system{B}'}$ be a random private state as defined by \Cref{eq:def-random-gamma}. 
    Then, there exist universal constants $\probconstant,\upperconstant>0$ such that
    \begin{align*} 
    & \probability(D_{\sep(\classic{A}\classic{A}'{:}\classic{B}\classic{B}')}\left(\gamma
    \|\sepstates(\system{A}\system{A}'{:}\system{B}\system{B}')\right)\leq\frac{\upperconstant}{\sqrt{d}}) \geq 1 - e^{-{\probconstant}d} \, .
    \end{align*}
\end{theorem}

\begin{proof}
    We know from \Cref{prop:gamma-tau-SEP-order} that, with probability greater than $1-e^{-{\probconstant}d}$, \(\gamma \leq _{\sep} (1+\upperconstant/\sqrt{d}) \hat\gamma\).
    Because \(\hat\gamma\) is a separable state, we have by \Cref{lemma:order-relent} that, with probability greater than $1-e^{-{\probconstant}d}$,
    \[D_{\sep}(\gamma\parallel\sepstates) \leq 
    D_{\sep}\left(\gamma\parallel\hat{\gamma}\right)
    \leq \log\left(1+\frac{\upperconstant}{\sqrt{d}}\right) \leq \frac{\upperconstant}{\sqrt{d}}. \]
    This concludes the proof.
\end{proof}

\section{Upper bounds on the quantum key repeater rate of private quantum states}
\label{appendix:R_D}

In this appendix, we come back to the more usual quantum key repeater setting, where Charlie can act jointly on arbitrary many copies of the input states, and not just one as we were imposing in \Cref{sec:repeater}. We first establish an upper bound on the corresponding quantum key repeater rate, which improves on previously known upper bounds. We then turn to looking at how to interpret this quantity in the case of private states.

\subsection{Upper bound on the quantum key repeater rate} 

Let us first recall the definition of the highest rate for the single node quantum key repeater.

\begin{definition}[Quantum key repeater rate~\cite{BCHW}] For any states $\rhoAsystems$ and $\rhoBsystems$, we define
    \label{def:repeaterrate}
    \begin{align*}
    R_D (\rhoA,\rhoB) \coloneqq &
    \lim _{\varepsilon\to0} \lim _{n\to\infty}
    \sup _{\Lambda}
    \quantity{R: \Lambda\quantity(\rhoA^{\otimes n}\otimes\rhoB^{\otimes n})
        \approx_\varepsilon \gamma^{Rn}}
    \end{align*}
    where $\Lambda\in\locc
    (\quantum{A}^n{:}\classic{\charlieA}^n\classic{\charlieB}^n{:}\quantum{B}^n)$.
\end{definition}

The following upper bound on $R_D(\rhoA,\rhoB)$, in terms of a regularized LOCC-restricted relative entropy distance
to quadri-separable states, was derived in~\cite[Theorem 4]{BCHW}. 
First, given $\varrho$ a state on $\alice\otimes\charlieA\otimes\bob\otimes\charlieB$,  define for each $n\in\N$
\begin{align*} 
D&_{\locc
  (\classic{\alice}\classic{\bob}
{:}\classic{\charlieA}\classic{\charlieB})}^n 
\left( \varrho
\middle\| \sepstates(\alice{:}\charlieA{:}\charlieB{:}\bob) \right) 
\\&\coloneqq \frac{1}{n} D_{\locc(\classic{\alice}^n\classic{\bob}^n{:}\classic{\charlieA}^n\classic{\charlieB}^n)} 
\left( \varrho^{\otimes n} 
\middle\| \sepstates(\alice^n{:}\charlieA^n{:}\charlieB^n{:}\bob^n) \right)\, ,
\intertext{and then standard regularization, define}
D&_{\locc
  (\classic{\alice}\classic{\bob}
{:}\classic{\charlieA}\classic{\charlieB})}^{\infty} 
\left( \varrho
\middle\| \sepstates(\alice{:}\charlieA{:}\charlieB{:}\bob) \right) 
\\&\coloneqq \lim_{n\to\infty} D_{\locc
  (\classic{\alice}\classic{\bob}
{:}\classic{\charlieA}\classic{\charlieB})}^n 
\left( \varrho
\middle\| \sepstates(\alice{:}\charlieA{:}\charlieB{:}\bob) \right)\, .
\end{align*}
We can now state the theorem.

\begin{theorem}[Upper bound on $R_D$~\cite{BCHW}]
    \label{lemma:repeatersep}
    For any states $\rhoAsystems$ and $\rhoBsystems$, we have
    \begin{equation} \label{eq:originalbound} R_D(\rhoA,\rhoB) 
    \leq D_{\locc(\classic{\alice}\bob{:}\classic{\charlieA}\classic{\charlieB})}^\infty 
    \left( \rhoA\otimes\rhoB
    \middle\| \sepstates(\alice{:}\charlieA{:}\charlieB{:}\bob) \right)\, . 
    \end{equation}
\end{theorem}

We can improve this upper bound on the quantum key repeater rate with a simple observation. 
A crucial passage of the original proof is that 
no repeater protocol will ever output entanglement between Alice and Bob,
if the inputs are substituted with separable states, i.e.
\[ \sigma \in \sepstates(\alice{:}\charlieA{:}\charlieB{:}\bob)\ \Longrightarrow \Tr_{\charlieA\charlieB} \Lambda(\sigma) \in \sepstates(\alice{:}\bob)\, . \]
However, only separability of the output is used, which holds already if we substitute a single input, instead of both,
with a separable state, i.e.
\[  \sigma \in \sepstates(\alice{:}\charlieA)\ \Longrightarrow \Tr_{\charlieA\charlieB} \Lambda(\sigma\otimes\rhoB ) \in \sepstates(\alice{:}\bob)\] \text{and}\[ \sigma \in \sepstates(\charlieB{:}\bob)\ \Longrightarrow \Tr_{\charlieA\charlieB} \Lambda(\rhoA\otimes\sigma) \in \sepstates(\alice{:}\bob)\, . \]

This works more generally for any set of states that gets mapped into the separable states $\sepstates(\alice{:}\bob)$ after tracing Charlie, as formalized by the following bound.

\begin{theorem} \label{th:R_D}
    Let $\states\in\allstates(\system{A\charlieA\charlieB B})$
    be a set of states such that%
    $\Lambda\in\locc(\classic{A}{:}\classic{\charlieA}\classic{\charlieB}{:}\classic{B})$ implies $\Lambda(\states)\subseteq\sepstates(\ab)$.
    Then, for any states $\rhoAsystems$ and $\rhoBsystems$, we have
    \begin{align*}
    R_D (\rhoA,\rhoB) 
    \leq  D^\infty_{\locc(\classic{A}{:}\classic{\charlieA}\classic{\charlieB}{:}\classic{B})} \quantity\big
    ( \rhoA\otimes\rhoB\| \states ).
    \end{align*}
\end{theorem}
    More precisely, the regularization requires to define the set of states $\states_n\in\allstates
    (\system{A}^n\system{\charlieA}^n\system{\charlieB}^n\system{B}^n)$ 
    for every $n$ (something that is implicit in $\sepstates(\alice^n{:}\bob^n)$ for separable states, and similarly for PPT states), so that 
    \begin{align*} 
    D^\infty_{\locc} \quantity( \rhoA\otimes\rhoB\| \states )
    & = \lim_{n\to\infty} \frac{1}{n}   D_{\locc_n}
    ^n \quantity( \rhoA\otimes\rhoB\| \states_n )
    \end{align*}
    where we omitted $\locc\equiv\locc
    (\classic{A}{:}\classic{\charlieA}\classic{\charlieB}{:}\classic{B})$
    and 
    $\locc_n\equiv\locc_n
    (\classic{A}^n{:}\classic{\charlieA}^n\classic{\charlieB}^n{:}\classic{B}^n)$

    In particular%
    \begin{equation} 
    \label{eq:betterbound} 
    R_D(\rhoA,\rhoB) \leq  D_{\locc+\rhoB}^\infty ( \rhoA\| \sepstates ), D_{\locc+\rhoA}^\infty ( \rhoB\| \sepstates ) 
    \end{equation}
    where $\sepstates\equiv\sepstates(\alice{:}\bob)$ and we defined 
    \begin{align*} 
    & D_{\locc+\rhoB}^\infty ( \rhoA\| \sepstates ) \coloneqq  
    \inf_{\mathllap\sigma \in \sepstates(\alice{:}\charlieA)} 
    D_{\locc(\classic{\alice}\classic{\bob}
          {:}\classic{\charlieA}\classic{\charlieB})}^\infty 
    \big( \rhoA\otimes\rhoB \big\| \sigma\otimes\rhoB \big), \\ 
    & D_{\locc+\rhoA}^\infty ( \rhoB\| \sepstates ) \coloneqq  
    \inf_{\mathllap\sigma \in \sepstates(\bob{:}\charlieB)} 
    D_{\locc(\classic{\alice}\classic{\bob}
          {:}\classic{\charlieA}\classic{\charlieB})}^\infty 
    \big( \rhoA\otimes\rhoB \big\| \rhoA\otimes\sigma \big).
    \end{align*} 

These bounds solve the ``factor of $2$'' issue about tightness of the original upper bound in some obvious simple cases.
For example, suppose that \(\rhoA = \rhoB = \psi^m\), where $\psi^m$ is the maximally entangled state on $\system{A}^{\otimes m}\system{B}^{\otimes m}$ (i.e.~equivalently, $\psi^m=\psi^{\otimes m}$ for $\psi$ the maximally entangled state on $\system{A}\system{B}$).
Then it is clear that there is equality in \Cref{eq:betterbound}, both sides being equal to $m$. However, the right hand side of \Cref{eq:originalbound} yields $2m$. Intuitively, while the original bound measures the distinguishability of both states
from separable, the new bounds measure the distinguishability of a single state and
consider the other one as an assisting resource to the measurement.

\subsection{Distinguishability of private quantum states from their key-attacked versions}
\label{app:pbit-vs-keyattacked}

We now compute upper bounds on the distinguishability of any private state of the form given by \Cref{eq:def-gamma} (not necessarily random) from its key-attacked version. Our measures of distinguishability are $\mathbf{M}$ relative entropy distances. These bounds are obtained by comparing this discrimination task to the one where the maximally entangled state in the key has been corrected and provided as a resource.

\begin{theorem}
    \label{prop:corrected-key}
    Let $\gamma$ be a private state on  ${\system{A}\system{A}'\system{B}\system{B}'}$ as defined by \Cref{eq:def-gamma}. Let $\shieldplusminus$ be the corresponding shield states on ${\system{A}'\system{B}'}$ and define $\shield \coloneqq (\shieldplus + \shieldminus)/2$. 
    Then, for any set of measurements $\measurements$ on $\system{A}\system{A}'\system{B}\system{B}'$, we have
    \begin{align}
    \label{lemma:pbithiding:singlecopy}
    D _\measurements (\gamma \parallel \hat \gamma) 
    &\leq  \frac{1}{2}  D _{\measurements+\psi}( \shieldplus \parallel \shield) 
    +\frac{1}{2}  D _{\measurements+\psi}( \shieldminus \parallel \shield) \, ,
    \\
    \label{lemma:pbithiding:regularized}
    D ^\infty _\measurements (\gamma \parallel \hat \gamma) 
    &\leq      \frac{1}{2}
    D ^\infty _{\measurements+\psi^2} \left( \shieldplus\otimes\shieldminus \parallel \shield\otimes\shield \right)\, . 
    \end{align}
    where
    \begin{align*}
    D        _{\measurements+\psi}( \alpha \parallel \beta) &\coloneqq 
    D        _\measurements (\psi \otimes \alpha \parallel \psi \otimes \beta)
    \\
    D^\infty _{\measurements+\psi^n}( \alpha \parallel \beta) &\coloneqq 
    D^\infty _\measurements (\psi^{\otimes n} \otimes \alpha \parallel \psi^{\otimes n} \otimes \beta)
    \;.
    \end{align*}  
\end{theorem}

\begin{proof}
    \Cref{lemma:pbithiding:singlecopy} follows straightforwardly from joint convexity of the relative entropy.
    \begin{align*} 
    D_\measurements(\gamma\parallel\hat\gamma) 
    & \leq \frac{1}{2} D_\measurements \left(\psi^+ \otimes \shieldplus  \parallel \psi^+ \otimes  \shield \right) 
    \\&\quad+
    \frac{1}{2} D_\measurements \left(\psi^- \otimes \shieldminus \parallel \psi^- \otimes  \shield \right) \\
    & = \frac{1}{2} D_\measurements \left(\psi   \otimes \shieldplus  \parallel \psi \otimes  \shield \right) 
    \\&\quad+
    \frac{1}{2} D_\measurements \left(\psi   \otimes \shieldminus \parallel \psi \otimes  \shield \right) \, .
    \end{align*}
    Indeed, the first inequality is by joint convexity of the restricted relative entropy. The last equality is by correcting the phase flips via a reversible local unitary, which can be done because the maximally entangled states are not in a mixture any more. 

    For all $n\in\N$ we have
    \begin{align*} 
    D&_\measurements^n(\gamma\parallel\hat\gamma) =
    \frac{1}{n} 
    D_\measurements(\gamma^{\otimes n} \| \hat \gamma ^{\otimes n})
    \\& \leq 
    \frac{1}{n} \frac{1}{2^{n}} \sum_{x_1,\mathrlap{\dots ,x_n=\pm}}
    D_\measurements(
    \psi^{n} \otimes \shield^{x_1} \otimes \dots \otimes \shield^{x_n} \| 
    \psi^{n} \otimes \shield^{\otimes n}) 
    \,,
    \end{align*}
    which we rewrite as 
    \begin{align*} 
    D_\measurements^n(\gamma\parallel\hat\gamma) 
    \leq 
    \frac{1}{n} \frac{1}{2^{n}} \sum_{\mathbf{x}\in\{\pm\}^n}
    D_\measurements(
    \psi^{\otimes n}\otimes{\shield}^{\otimes \mathbf{x}} \|
    \psi^{\otimes n}\otimes{\shield}^{\otimes n})
    \,.
    \end{align*}
    Notice that because $\shieldplusminus$ are orthogonal, we have $D(\shieldplusminus\|\shield)=1$, and thus for any $\mathbf{x}$ we have 
    \begin{align*} 
    D_\measurements(&
    \psi^{\otimes n}\otimes{\shield}^{\otimes \mathbf{x}} \|
    \psi^{\otimes n}\otimes{\shield}^{\otimes n}) 
    \\&\leq D(
    \psi^{\otimes n}\otimes{\shield}^{\otimes \mathbf{x}} \|
    \psi^{\otimes n}\otimes{\shield}^{\otimes n}) 
    \\&= D(
                           {\shield}^{\otimes \mathbf{x}} \|
                           {\shield}^{\otimes n}) 
    \\&= \sum\limits_{k=1}^n D({\shield}^{{x}_k} \| {\shield})
    \\&= n
    \,.
    \end{align*}
    \newcommand {\typical}{         {T}(n,\epsilon)}
    \newcommand{\atypical}{\overline{T}(n,\epsilon)}
    If we now denote with $\typical$ the set of $\epsilon$-typical sequences and with $\atypical$ its complement, we thus find that 
    \begin{align*}
    \frac{1}{n} \frac{1}{2^{n}} &\sum_{\mathbf{x}\in \mathrlap\atypical}
        D_\measurements(
        \psi^{\otimes n}\otimes{\shield}^{\otimes \mathbf{x}} \|
        \psi^{\otimes n}\otimes{\shield}^{\otimes n})
    \\&\leq
        \sum_{\mathbf{x}\in \mathrlap\atypical} 2^{-n} 
        = \probability(\atypical)
    \end{align*}
    Fixed $\epsilon>0$, for $n$ is large enough we have $\probability(\atypical)\leq \epsilon$, and thus
    \begin{align*} 
    D_\measurements^n(\gamma\parallel\hat\gamma) 
      &\leq 
    \frac{1}{n} \frac{1}{2^{n}} \sum_{\mathbf{x}\in \mathrlap\typical}
    D_\measurements(
    \psi^{\otimes n}\otimes{\shield}^{\otimes \mathbf{x}} \|
    \psi^{\otimes n}\otimes{\shield}^{\otimes n})
    \\& + 
    \frac{1}{n} \frac{1}{2^{n}} \sum_{\mathbf{x}\in \mathrlap\atypical}
    D_\measurements(
    \psi^{\otimes n}\otimes{\shield}^{\otimes \mathbf{x}} \|
    \psi^{\otimes n}\otimes{\shield}^{\otimes n})
    \\&\leq
    \frac{1}{n} \max_{\mathbf{x}\in \typical}
    D_\measurements(
    \psi^{\otimes n}\otimes{\shield}^{\otimes \mathbf{x}} \|
    \psi^{\otimes n}\otimes{\shield}^{\otimes n}) + \epsilon
    \,.
    \end{align*}
    
    Thus we can take $\epsilon =1/n$ and in the limit for large $n$ we find
    \begin{align*} 
    D_\measurements^\infty(\gamma\parallel\hat\gamma) 
      &=\lim_{n\to\infty} D_\measurements^n(\gamma\parallel\hat\gamma) 
    \\&\leq \lim_{n\to\infty}
    \frac{1}{n} \frac{1}{2^{n}} \sum_{\mathbf{x}\in \mathrlap{{T}(n,1/n)}}
    D_\measurements(
    \psi^{\otimes n}\otimes{\shield}^{\otimes \mathbf{x}} \|
    \psi^{\otimes n}\otimes{\shield}^{\otimes n})
    \,.
    \end{align*}

\end{proof}

The interpretation of the quantities appearing on the right hand side of \Cref{lemma:pbithiding:singlecopy} and \Cref{lemma:pbithiding:regularized} goes as follows. In the single-copy version, we have one bit of pure entanglement as a resource, that we can use to distinguish $\shieldplus$ from $\shieldminus$. To write this concisely we introduced the notation \(\measurements+\psi\), which stands for maps in \(\measurements\) assisted by one bit of entanglement for each state. In the regularized version, we are allowing two bits of entanglement for each copy of $\shieldplus\otimes\shieldminus$. Note that we can argue that this upper bound is always finite and at most $1$, because it is itself upper bounded by the global relative entropy. In the other direction, it can be shown to be always at least $2/\log d$, because with two bits of entanglement, Alice and Bob can perform global measurements on a $2/\log d$ share of the copies of $\shieldplus$ and $\shieldminus$. Hence summing up,
\[ \frac{2}{\log d} \leq D^\infty_{\measurements+\psi^2} 
\left(\shieldplus\otimes\shieldminus 
\middle\| \shield\otimes\shield \right) \leq 2\, . \]


\begin{thebibliography}{99}
\newcommand{\bibinfo}[2]{#2}
\newcommand{\bibauthor}[2]{\bibinfo{author}{\textbf{\StrLeft{#1}{1}.~{#2}}}}
\newcommand{\bibtitle}[1]{\emph{#1}}
\newcommand{\bibauthors}[1]{{#1}}
\newcommand{\bibjournal}[1]{{#1}}

\newcommand{\printurl}{http://dx.doi.org/}


\newcommand{\journalplain}[4]{\emph{\bibinfo{journal}{#1}},\ {\bibinfo{volume}{#2}}:\bibinfo{pages}{#3},\ \bibinfo{year}{#4}}
\newcommand{\proceedingplain}[2]{\emph{\bibinfo{journal}{#1}},\ \bibinfo{pages}{#2}}
\newcommand{\journal}   [5]{\href{\printurl#5}{\journalplain{#1}{#2}{#3}{#4}}}
\newcommand{\proceeding}[3]{\href{\printurl#3}{\proceedingplain{#1}{#2}}}

\NewDocumentCommand{\arXiv}{g d:.}
{%
    \IfNoValueTF{#1}%
        {\IfNoValueF{#2}%
            {arXiv:\href{http://arxiv.org/abs/#2}{#2}.}%
        }%
        {arXiv:\href{http://arxiv.org/abs/#1}{#1}}%
}

\bibitem{HLW} \bibauthors{P.~Hayden, D.~Leung, A.~Winter}, \bibtitle{Aspects of generic entanglement}, \bibjournal{Commun.~Math.~Phys.}, 265(1):95--117, 2006. 

\bibitem{ASY} \bibauthors{G.~Aubrun, S.J.~Szarek, D.~Ye}, \bibtitle{Entanglement thresholds for random induced states}, \bibjournal{Comm.~Pure App.~Math.}, 67(1):129--171, 2013. 

\bibitem{AL} \bibauthors{G.~Aubrun, C.~Lancien}, \bibtitle{Locally restricted measurements on a multipartite quantum system:~data hiding is generic}, \bibjournal{Quant.~Inf.~Comput.}, 15(5--6):512--540, 2014. 

\bibitem{HHHO} \bibauthors{K.~Horodecki, M.~Horodecki, P.~Horodecki, J.~Oppenheim}, \bibtitle{Secure key from bound entanglement}, \bibjournal{Phys.~Rev.~Lett.}, 94(160502), 2005. 

\bibitem{HHHOb} \bibauthors{K.~Horodecki, M.~Horodecki, P.~Horodecki, J.~Oppenheim}, \bibtitle{General paradigm for distilling classical key from quantum states}, \bibjournal{IEEE~Trans.~Inf.~Theory}, 55(1898), 2009. 


\bibitem{DVLT} \bibauthors{D.P.~DiVincenzo, D.W.~Leung, B.M.~Terhal}, \bibtitle{Hiding bits in Bell states}, \bibjournal{Phys.~Rev.~Lett.}, 86(25):5807--5810 (2001). 

\bibitem{BCHW} \bibauthors{S.~B\"{a}uml, M.~Christandl, K.~Horodecki, A.~Winter}, \bibtitle{Limitations on quantum key repeaters}, \bibjournal{Nature Commun.}, 6(6908), 2014.

\bibitem{MWW} \bibauthors{W.~Matthews, S.~Wehner, A.~Winter}, \bibtitle{Distinguishability of quantum states under restricted families of measurements with an application to data hiding}, \bibjournal{Commun.~Math.~Phys.}, 291(3):813--843, 2009. 

\bibitem{Piani} \bibauthors{M.~Piani}, \bibtitle{Relative entropy of entanglement and restricted measurements}, \bibjournal{Phys.~Rev.~Lett.}, 103(160504), 2009. 

\bibitem{berta2016variational} \bibauthors{M.~Berta, O.~Fawzi, M.~Tomamichel}, \bibtitle{On variational expressions for quantum relative entropies}, \bibjournal{Proc.~ISIT}, 2844, 2016. 


\bibitem{CF} \bibauthors{M.~Christandl, R~Ferrara}, \bibtitle{Private states, quantum data hiding and the swapping of perfect secrecy}, \bibjournal{Phys. Rev. Lett.} 119:220506, 2017. 

\bibitem{DH} \bibauthors{M.J.~Donald, M.~Horodecki}, \bibtitle{Continuity of Relative Entropy of Entanglement}, \bibjournal{Physics Letters A}, 264(4):257--260 (1999). 

\bibitem{Fannes} \bibauthors{M.~Fannes}, \bibtitle{A continuity property of the entropy density for spin lattice systems}, \bibjournal{Commun.~Math.~Phys.}, 31(4):291--294, 1973.

\bibitem{Audenaert} \bibauthors{K.M.R.~Audenaert}, \bibtitle{A sharp continuity estimate for the von Neumann entropy}, \bibjournal{J.~Phys.~A: Math.~Theor.} 40(28):8127--8136, 2007. 

\bibitem{Petz} \bibauthors{D.~Petz}, \bibjournal{Quantum Information Theory and Quantum Statistics}, Springer Verlag, Berlin Heidelberg, 2008.

\bibitem{Winter} \bibauthors{A.~Winter}, \bibtitle{Tight uniform continuity bounds for quantum entropies:~conditional entropy, relative entropy distance and energy constraints}, \bibjournal{Commun.~Math.~Phys.}, 347(1):291--313, 2016. 

\bibitem{KR} \bibauthors{I. Kim, M.-B. Ruskai}, \bibtitle{Bounds on the concavity of entropy}, \bibjournal{J. Math. Phys.}, 55:092201, 2014.

\bibitem{VPRK} \bibauthors{P.L.~Knight, M.B.~Plenio, M.A.~Rippin, V.~Vedral}, \bibtitle{Quantifying entanglement}, \bibjournal{Phys.~Rev.~Lett.}, 78(2275), 1997.

\bibitem{rains1999rigorous} \bibauthors{M.E.~Rains}, \bibtitle{A rigorous treatment of distillable entanglement}, \bibjournal{Phys.~Rev.~A}, 60(173), 1999. 


\bibitem{TV} \bibauthors{R.~Tarrach, G.~Vidal}, \bibtitle{Robustness of entanglement}, \bibjournal{Phys.~Rev.~A}, 59:141--155, 1999.

\bibitem{datta2009max} \bibauthors{N.~Datta}, \bibtitle{Max- Relative Entropy of Entanglement, alias Log Robustness}, \bibjournal{International Journal of Quantum Information}, 7:475, 2009. 

\bibitem{lami2017ultimate} \bibauthors{L.~Lami, C.~Palazuelos, A.~Winter}, \bibtitle{Ultimate data hiding in quantum mechanics and beyond}, arXiv:1703.03392[quant-ph]. 

\bibitem{EW} \bibauthors{T.~Eggeling, R.F.~Werner}, \bibtitle{Hiding classical data in multi-partite quantum states}, \bibjournal{Phys.~Rev.~Lett.}, 89(097905), 2002. 

\bibitem{HH} \bibauthors{M.~Horodecki, P.~Horodecki}, \bibtitle{Reduction criterion of separability and limits for a class of distillation protocols}, \bibjournal{Phys.~Rev.~A}, 59:4206, 1999.


\bibitem{BG} \bibauthors{H.~Barnum, L.~Gurvits}, \bibtitle{Largest separable balls around the maximally mixed bipartite quantum state}, \bibjournal{Phys.~Rev.~A}, 66(062311), 2002.

\bibitem{MM} \bibauthors{E.~Meckes, M.~Meckes}, \bibtitle{Spectral measures of powers of random matrices}, \bibjournal{Electron.~Commun.~Probab.}, 18(78):1--13, 2013. 

\bibitem{horodecki2016irreducible} \bibauthors{K.~Horodecki, P.~{\'C}wikli{\'n}ski, A.~Rutkowski, M.~Studzi{\'n}ski}, \bibtitle{Irreducible private states}, arXiv:1612.08938[quant-ph]. 

\bibitem{ozols2014bound} \bibauthors{M.~Ozols, G.~Smith, J.~Smolin}, \bibtitle{Bound entangled states with a private key and their classical counterpart}, \bibjournal{Phys.~Rev.~Lett.} 112:110502, 2014. 

\bibitem{Kretschmann} \bibauthors{D.~Kretschmann, R.~F.~Werner}, \bibtitle{Tema con variazioni:~quantum channel capacity}, \bibjournal{New Journal of Physics}, 6(1):26, 2004. 


\bibitem{Shirokov} \bibauthors{M.~E.~Shirokov}, \bibtitle{Tight continuity bounds for the quantum conditional mutual information, for the Holevo quantity and for capacities of quantum channels}, arXiv:1512.09047[quant-ph].

\bibitem{PPT^2} \bibauthors{M.~Christandl}, \bibtitle{PPT square conjecture}, \bibjournal{BIRS workshop: Operator structures in quantum information theory}, Problem G, 2012.  https://www.birs.ca/workshops/2012/12w5084/report12w5084.pdf

\bibitem{BCDL} \bibauthors{D.~Bru\ss{}, J.I.~Cirac, W.~D\"{u}r, M.~Lewenstein}, \bibtitle{Distillability and partial transposition in bipartite systems}, \bibjournal{Phys.~Rev.~A}, 61(0262313), 2000. 

\bibitem{DSSTT} \bibauthors{D.P.~DiVincenzo, P.W.~Shor, J.A.~Smolin, B.M.~Terhal, A.V.~Thapliyal}, \bibtitle{Evidence for bound entangled states with negative partial transpose}, \bibjournal{Phys.~Rev.~A}, 61(0262312), 2000.

\bibitem{LiW} \bibauthors{K.~Li, A.~Winter}, \bibtitle{Relative entropy and squashed entanglement}, \bibjournal{Commun.~Math.~Phys.},  326(1):63--80, 2014. 


\bibitem{Montanaro} \bibauthors{A.~Montanaro}, \bibtitle{Weak multiplicativity for random quantum channels}, \bibjournal{Commun.~Math.~Phys.}, 319(2):535--555, 2013. 

\bibitem{DCEL} \bibauthors{C.~Dankert, R.~Cleve, J.~Emerson, E.~Livine},
\bibtitle{Exact and approximate unitary 2-designs and their application to fidelity estimation}, \bibjournal{Phys.~Rev.~A}, 80(012304), 2009.

\bibitem{AS} \bibauthors{G.~Aubrun, S.J.~Szarek}, \bibtitle{Tensor product of convex sets and the volume of separable states on $N$ qudits}, \bibjournal{Phys.~Rev.~A}, 73(022109), 2006. 

\bibitem{CRT} \bibauthors{R.~Colbeck, R.~Renner, M.~Tomamichel}, \bibtitle{A fully quantum asymptotic equipartition property}, \bibjournal{IEEE Trans.~on Inf.~Theory}, 15(12):5840--5847, 2009. 


\bibitem{Bhatia} \bibauthors{R.~Bhatia}, \bibjournal{Matrix Analysis}, Graduate Texts in Mathematics Vol.~169. Springer-Verlag, New-York, 1997.

\bibitem{AGZ} \bibauthors{G.W.~Anderson, A.~Guionnet, O.~Zeitouni}, \bibjournal{An Introduction to Random Matrices}, Cambridge Studies in Advanced Mathematics, Vol.~118, Cambridge University Press, Cambridge, 2010.


\end{thebibliography}
\end{document}